\documentclass[12pt, a4paper]{article}
\usepackage[utf8]{inputenc}
\usepackage{dcolumn,lscape}
\usepackage{amsmath,longtable,multicol,dcolumn,tabularx,graphicx,amssymb}
\usepackage{exscale,amsthm,multirow,rotating}
\usepackage{natbib}
\usepackage{bbm,listings}
\usepackage{endfloat}
\usepackage{tikz,dsfont,float}
\usepackage[caption = false]{subfig}
\newcommand{\xvec}{\boldsymbol}
\newcommand{\xmat}{\mathbf}
\newcommand{\xset}{\mathds}

\newtheorem{theorem}{Theorem}

\newtheorem{corollary}{Corollary}
\newtheorem{lemma}{Lemma}

\newcommand{\pkg}[1]{\textbf{#1}}
\newcommand{\code}[1]{\texttt{#1}}
\newcommand{\class}[1]{`\code{#1}'}
\newcommand{\fct}[1]{\code{#1()}}
\newcommand{\proglang}[1]{\textsf{\textbf{#1}}}

\definecolor{lgrey}{rgb}{0.8,0.8,0.8}

\author{Philipp Otto\\Leibniz University Hannover, Germany}

\title{spGARCH: An R-Package for Spatial and Spatiotemporal ARCH models}

\begin{document}

\lstdefinelanguage{Renhanced}%
  {keywords={},%
   otherkeywords={},%
   alsoother={._$},%
   sensitive,%
   morecomment=[l]\#,%
   morestring=[d]",%
   morestring=[d]'
  }%
\lstset{language=Renhanced, numbers=left, numberstyle=\tiny\color{lgrey}, basicstyle=\small\ttfamily, keywordstyle=\small,
showstringspaces=false, tabsize=2, commentstyle=\color{code_comment}, extendedchars=true, breaklines=true, frame=single, morecomment=[l][\color{red}]{R> },
escapeinside={\%*}{*)}}

\maketitle

\begin{abstract}
 In this paper, a general overview on spatial and spatiotemporal ARCH models is provided. In particular, we distinguish between three different spatial ARCH-type models. In addition to the original definition of \cite{Otto16_arxiv}, we introduce an exponential spatial ARCH model in this paper. For this new model, maximum-likelihood estimators for the parameters are proposed. In addition, we consider a new complex-valued definition of the spatial ARCH process. From a practical point of view, the use of the \proglang{R}-package \pkg{spGARCH} is demonstrated. To be precise, we show how the proposed spatial ARCH models can be simulated and summarize the variety of spatial models, which can be estimated by the estimation functions provided in the package. Eventually, we apply all procedures to a real-data example.
\end{abstract}



\section[Introduction]{Introduction} \label{sec:intro}

Whereas autoregressive conditional heteroscedasticity (ARCH) models are applied widely in time series analysis, especially in financial econometrics, spatial conditional heteroscedasticity has not been seen as critical issue in spatial econometrics up to now. Although it is well-known that classical least squares estimators are biased for spatially correlated data as well as for spatial data with an inhomogeneous variance across space, there are just a few papers proposing statistical models accounting for spatial conditional heteroscedasticity in terms of the ARCH and GARCH models of \cite{Engle82} and \cite{Bollerslev86}. The first extensions to spatial models attempted were time series models incorporating spatial effects in temporal lags (see \citealt{Borovkova12} and \citealt{Caporin06}, for instance). Instantaneous spatial autoregressive dependence in the conditional second moments, i.e., the conditional variance in each spatial location is influenced by the variance nearby, has been introduced by \cite{Otto16_arxiv,Otto18_spARCH,Otto18_tech}. Their models allow for these instantaneous effects but require certain regularity conditions. In this paper, we propose an alternative specification of spatial autoregressive conditional heteroscedasticity based on an exponential definition of the conditional variance. This new model can be seen as the spatial equivalent of the exponential GARCH model by \cite{Nelson91}. Other recent papers propose a mixture of these two approaches (see \citealt{Sato17,Sato18b}). Moreover, all these models can be used in spatiotemporal settings (see \citealt{Otto16_arxiv,Sato18a}).

In addition to the novel spatial exponential ARCH model, this paper demonstrates the use of the \proglang{R}-package \pkg{spGARCH}. From this practical point of view, the simulation of several spatial ARCH-type models as well as the estimation of a variety of spatial models with conditional heteroscedasticity are shown. There are several packages implementing geostatistical models, kriging approaches, and other spatial models (cf. \citealt{Cressie93,Cressie11}). One of the most powerful packages used to deal with models of spatial dependence is \pkg{spdep}, written by \cite{spdep}. It implements most spatial models in a user-friendly way, such as spatial autoregressive models, spatial lag models, and so forth (see, also, \citealt{Elhorst10} for an overview). These models are typically called spatial econometrics models, although they are not tied to applications in economics. In contrast, the package \pkg{gstat} provides functions for geostatistical models, variogram estimation, and various kriging approaches (see \citealt{gstat} for details). For dealing with big geospatial data, the \pkg{Stem} package uses an expectation-maximization (EM) algorithm for fitting hierarchical spatiotemporal models (see \citealt{Stem} for details). For a distributed computing environment, the \proglang{MATLAB} software \pkg{D-STEM} from \cite{Finazzi14} also provides powerful tools for dealing with heterogeneous spatial supports, large multivariate data sets, and heterogeneous spatial sampling networks. Additionally, these fitted models are suitable for spatial imputation. Contrary to these EM approaches, Bayesian methods for modeling spatial data are implemented in the \proglang{R}-\pkg{INLA} package (see \citealt{inla} for technical details of the integrated nested Laplace approximations and \citealt{inla2} for recently implemented features). Along with this package, the \proglang{R}-\pkg{INLA} project provides several functions for diverse spatial models incorporating integrated nested Laplace approximations.

In contrast to the above mentioned software for spatial models, the prevalent \proglang{R}-package for time series GARCH-type models is \pkg{rugarch} from  \cite{rugarch}. Since \pkg{spGARCH} has been developed mainly to deal with spatial data, we aim to provide a package which is user-friendly for researchers and data scientists working in applied spatial science. Thus, the package is coordinated with the objects and ideas of \proglang{R} packages for spatial data rather than packages for dealing with time series.

We structure the paper as follows. In the next Section \ref{sec:models}, we discuss all covered spatial and spatiotemporal ARCH-type models. In addition, we introduce a novel exponential spatial ARCH model, which has weaker regularity conditions than the other spatial ARCH models. In the subsequent section, parameter estimation based on the maximum-likelihood principle is discussed for both the previously proposed spatial ARCH models as well as the new exponential spatial ARCH model. After these theoretical sections, we demonstrate the use of the \proglang{R}-package \pkg{spGARCH} in Section \ref{sec:spGARCH_package}. Further, we fit a spatial autoregressive model with exogenous regressors and spatial ARCH residuals for a real-world data set. In particular, we analyze prostate cancer incidence rates in southeastern U.S. states. Section \ref{sec:summary} concludes the paper.



\section{Spatial ARCH-type models} \label{sec:models}

Let   $\left\{Y(\xvec{s}) \in \xset{R}: \xvec{s} \in D \right\}$ be a univariate stochastic process having a spatial autoregressive structure in the conditional variance. The process is defined in a multidimensional space $D$, which is typically a subset of the $q$-dimensional real numbers $\xset{R}^q$, as space is usually finite. For dealing with spatial lattice data, $D$ is subset of the $q$-dimensional integers $\xset{Z}^q$. For both cases, it is important that the subset contains a $q$-dimensional rectangle of positive volume (cf. \citealt{Cressie11}). Moreover, this definition is suitable for modeling spatiotemporal data, as one might assume that $D$ is the product set $\xset{R}^{k} \times \xset{Z}^{l}$ with $k + l = d$.

To define spatial models, in particular areal spatial models such as the simultaneous autoregressive (SAR) models, it is convenient to consider a vector of observations \linebreak $\xvec{Y} = (Y(\xvec{s}_1), \ldots, Y(\xvec{s}_n))^\prime$ at all locations $\xvec{s}_1, \ldots, \xvec{s}_n$. For spatial ARCH models, we specify this vector as
\begin{equation}\label{eq:spARCHy}
\xvec{Y} = \mathrm{diag}(\xvec{h})^{1/2}\xvec{\varepsilon} \, ,
\end{equation}
an analogue to the well-known time series ARCH models (cf. \citealt{Engle82,Bollerslev86}). However, note that the vector $\xvec{h}$ does not necessarily coincide with the conditional variance
\begin{equation*}
Var(Y(\xvec{s}_i) | Y(\xvec{s}_1), \ldots, Y(\xvec{s}_{i-1})) \, ,
\end{equation*}
as the variance in any location $\xvec{s}_j$ also depends on $Y(\xvec{s}_i)$ for $j \neq i$ (see \citealt{Otto16_arxiv} for details).
We now distinguish between several spatial ARCH-type models via the definition of $\xvec{h}$.

\subsection{Spatial ARCH model}\label{sec:spARCH}

First, we define this vector $\xvec{h}$ in such a way as to be analogous to the definition in \cite{Otto16_arxiv,Otto18_spARCH}. For this model, the vector $\xvec{h}_O$ is given by
\begin{equation}\label{eq:h_spARCH}
\xvec{h}_O = \alpha\xvec{1} + \rho\xmat{W} \mathrm{diag}(\xvec{Y}) \xvec{Y} \, ,
\end{equation}
where $\mathrm{diag}(\xvec{a})$ is a diagonal matrix with the entries of $\xvec{a}$ on the diagonal. In order to be consistent with the implementation in the \proglang{R}-package \pkg{spGARCH}, we focus on the special case with two parameters $\alpha$ and $\rho$, whereas \cite{Otto16_arxiv} proposed a more general model with a vector $\xvec{\alpha} = (\alpha_1, \ldots, \alpha_n)'$ and the first-order spatial lag $\xmat{W}\mathrm{diag}(\xvec{Y}) \xvec{Y}$.

For this definition, there is a one-to-one relation between $\xvec{Y}$ and $\xvec{\varepsilon}$ via the squared observations $\xvec{Y}^{(2)} = (Y(\xvec{s}_1)^2, \ldots, Y(\xvec{s}_n)^2)^\prime$ and squared errors $\xvec{\varepsilon}^{(2)} = (\varepsilon(\xvec{s}_1)^2, \ldots, \varepsilon(\xvec{s}_n)^2)^\prime$ with
\begin{equation}\label{eq:Y2}
\xvec{Y}^{(2)} = \alpha \, \left(\xmat{I} - \xmat{A} \right)^{-1} \xvec{\varepsilon}^{(2)}\, ,
\end{equation}
where $\xmat{W}$ is a predefined spatial weighting matrix and
\begin{equation*}\label{eq:A}
\xmat{A} = \rho \, \mathrm{diag}\left(\varepsilon(\xvec{s}_1)^2, \ldots, \varepsilon(\xvec{s}_n)^2\right) \xmat{W} \, .
\end{equation*}
Thus,
\begin{equation*}
\xvec{h}_O = \alpha\xvec{1} + \rho \alpha \xmat{W} \left(\xmat{I} - \xmat{A} \right)^{-1} \xvec{\varepsilon}^{(2)} \, .
\end{equation*}


It is important to assume that the spatial weighting matrix is a non-stochastic, positive matrix with zeros on the main diagonal to ensure that a location is not influenced by itself (cf. \citealt{Elhorst10,Cressie11}). The vector of random errors is denoted by $\xvec{\varepsilon}$. Due to the complex dependence implied by the weighting matrix $\xmat{W}$, $\xvec{h}_O$ is not necessarily positive; thus, $\mathrm{diag}(\xvec{h})^{1/2}$ does not necessarily have a solution in the real numbers such that the process in \eqref{eq:spARCHy} is well-defined. This is only the case if the condition of the following lemma is fulfilled.

\begin{lemma}[\citealt{Otto16_arxiv}]\label{th:spARCH}
Suppose that $\alpha \geq 0$, $\rho \geq 0$ and that $\det(\xmat{I} - \xmat{A}^2) \neq 0$.
If all elements of the matrix
\begin{equation}
(\xmat{I} - \xmat{A}^2)^{-1} \label{eq:inverse_I-A}
\end{equation}
are nonnegative, then all components of $\xvec{Y}^{(2)}$ are nonnegative, i.e., $Y(\xvec{s}_i)^2 \ge 0$ for $i=1,\ldots,n$. Moreover, $h_O(\xvec{s}_i) \ge 0$ for $i=1,\ldots,n$.
\end{lemma}

It is important to note that $\xmat{A}$ depends on both the weighting matrix and the realizations of the errors. In order to ensure that this condition is fulfilled, \cite{Otto16_arxiv} propose to truncate the support of the error distribution on the interval $(-a, a)$ with
\begin{equation*}
  a = \left\{
  \begin{array}{cc}
  \infty                                       &  \exists k > 0 : \rho\xmat{W}^k = \xmat{0} \\ 
  {1}/{\sqrt[4]{ \rho^2||\xmat{W}^2 ||_1}}     &  \rho^2 || \xmat{W}^2 ||_1 > 0
  \end{array} \right. \, ,
\end{equation*}
where $|| \cdot ||_1$ denotes the matrix norm based on the $l_1$ vector norm.

There are two cases in which the support of the errors does not need to be constrained. If $\rho = 0$, the process coincides with a spatial white noise process such that $a$ equals $\infty$. Moreover, all entries of $\xvec{h}$ are non-negative if $\xmat{W}$ is similar to a strictly triangular matrix. Then, $\xmat{W}$ is nilpotent. This case covers the classical time-series ARCH($p$) models introduced by \cite{Engle82} as well as the so-called oriented spARCH processes. For these processes, the spatial dependence has a certain direction, e.g., observations are only influenced by observations in a southward direction or by observations which are closer to an arbitrarily chosen center. The setting also covers recent time-series GARCH models incorporating spatial information (e.g., \citealt{Borovkova12,Caporin06}).

Of course, the truncated support of the errors has an impact on the extent of the spatial dependence on the conditional variances. Obviously, the support need not be constrained regarding $\rho = 0$. However, this support decreases with increasing values of $\rho$. For instance, if $\rho = 1$, then the parameter $a$ is equal to $0.968$ for Rook's contiguity matrices on a two-dimensional lattice. As a measure of the spatial dependence of the variance, one might consider Moran's $I$ for the squared observations (see \citealt{Moran50}). Moreover, we observe that the growth rate of $I$ decreases with increasing spatial weights. This trend can be explained by the compact support of the errors. Since there cannot be large variations $\varepsilon(\xvec{s}_i)$ in absolute terms, there also cannot be large spatial clusters of high or low variance. To illustrate this behavior, Figure \ref{fig:MoransI} depicts Moran's $I$ for simulated observations $\xvec{Y}$ and their squares for $\rho \in \{0,0.05, \ldots, 2\}$. For the Monte Carlo simulation study, we simulate $n = 400$ observation on a two-dimensional lattice $D = \{\xvec{s} = (s_1, s_2)' \in \xset{Z}^2: 0 \leq s_1, s_2 \leq 20\}$. The weighting matrix is a common Rook's contiguity matrix, and the simulation is done for $10^5$ replications. Although the exact distribution of Moran's statistic is bounded, the standardized statistic is asymptotically normally distributed for the ``majority of spatial structures'' (\citealt{Tiefelsdorf95}, see also \citealt{Cliff81}). Thus, the asymptotic 95\% confidence intervals are plotted in Figure \ref{fig:MoransI}, as well.

\begin{figure}
  \centering
  \includegraphics[width=0.45\textwidth]{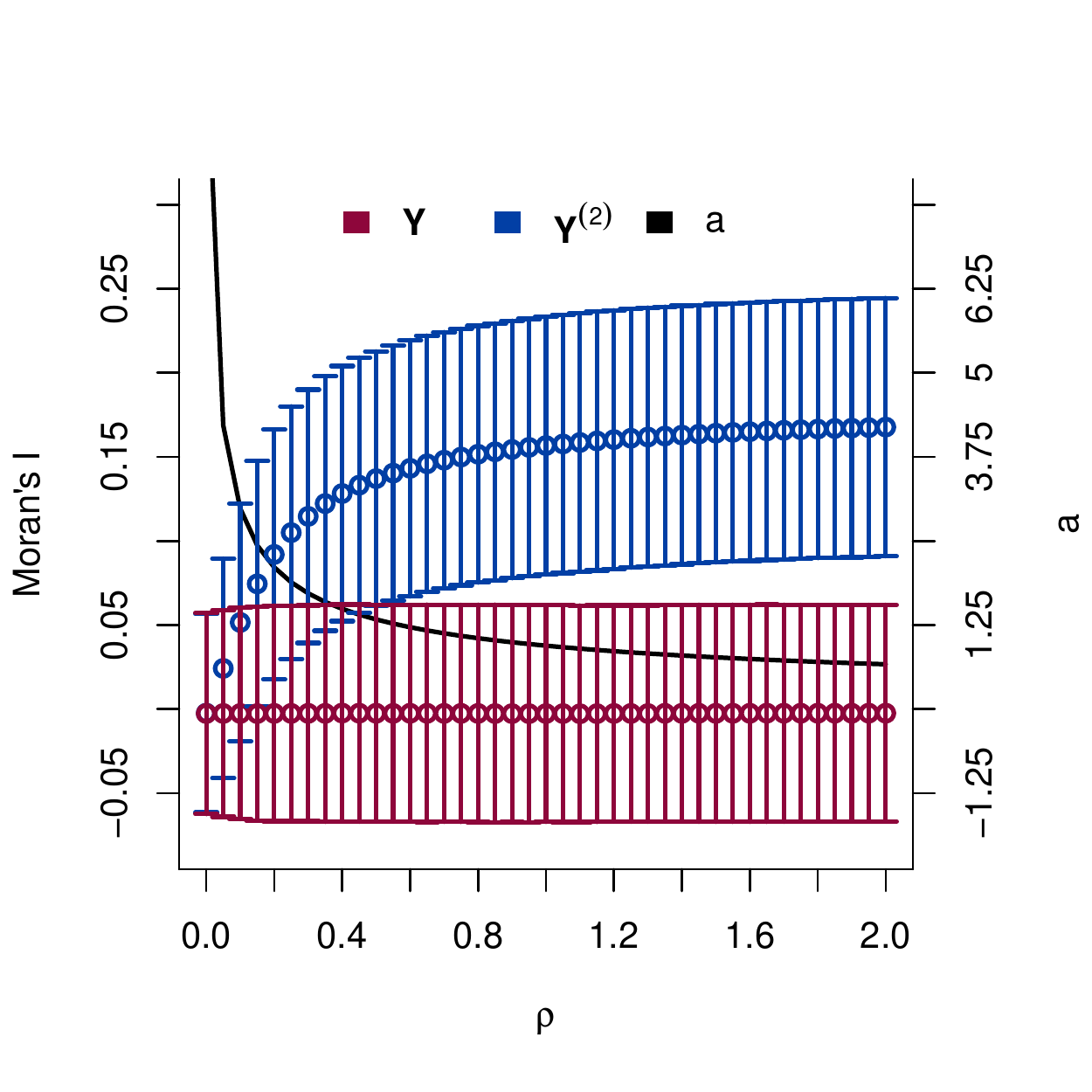}
  \caption{Moran's $I$ of the observations $\xvec{Y}$ and the squared observations $\xvec{Y}^{(2)}$, including the asymptotic 95\% confidence intervals of $I$ for $\rho \in \{0,0.05, \ldots, 2\}$. The resulting bound $a$ is plotted as a bold, black line.}\label{fig:MoransI} 
\end{figure}

\subsection{Exponential Spatial ARCH model}\label{sec:EspARCH}

Next, we consider an exponential spatial ARCH process (E-spARCH). In this setting, we define the natural logarithm of $\xvec{h}_E = (h_{E}(\xvec{s}_1), \ldots, h_{E}(\xvec{s}_n))'$ as
\begin{equation}\label{eq:EspARCHh}
\ln \xvec{h}_E = \alpha\xvec{1} + \rho\xmat{W} g_b(\xvec{\varepsilon}) \, ,
\end{equation}
with a function $g_b: \xset{R}^n \rightarrow \xset{R}^n$.
Like \cite{Nelson91}, we assume that
\begin{equation*}
g_b(\xvec{\varepsilon}) = (\ln |\varepsilon(\xvec{s}_1)|^b, \ldots, \ln |\varepsilon(\xvec{s}_n)|^b)'
\end{equation*}
for positive values of $b$. For this definition, there is a one-to-one relation between $\xvec{Y}$ and $\xvec{\varepsilon}$, as we show in the following theorem.

\begin{theorem}\label{th:EspARCH}
Suppose that $\alpha > 0$, $\rho \geq 0$, and $w_{ij} \geq 0$ for all $i,j = 1, \ldots, n$ and $g_b(\xvec{\varepsilon}) = (\ln |\varepsilon(\xvec{s}_1)|^b, \ldots, \ln |\varepsilon(\xvec{s}_n)|^b)'$. Then there exists one and only one $Y(\xvec{s}_1), \ldots, Y(\xvec{s}_n)$ that corresponds to each $\varepsilon(\xvec{s}_1), \ldots, \varepsilon(\xvec{s}_n)$ for $b > 0$.
\end{theorem}
At location $\xvec{s}_i$, the value of $h_{E}(\xvec{s}_i)$ is then given by
\begin{equation*}
\ln h_{E}(\xvec{s}_i) = \alpha + \sum_{v=1}^{n} \rho b w_{iv} \ln |\varepsilon(\xvec{s}_v)| \; \text{for} \; i = 1, \ldots, n\, .
\end{equation*}
For this definition of $g_b$, one could rewrite $\ln \xvec{h}$ as
\begin{equation}\label{eq:h_EspARCH}
\ln \xvec{h}_E = \xmat{S} \left(\alpha\xvec{1} + \rho b \xmat{W} \ln|\xvec{Y}|\right)
\end{equation}
with
\begin{equation*}
\xmat{S} = (s_{ij})_{i,j = 1, \ldots, n} = \left(\xmat{I} + \frac{1}{2} \rho b \xmat{W}\right)^{-1} \, .
\end{equation*}
In contrast to the spARCH process described in Section \ref{sec:spARCH}, Corollary \ref{cor:EspARCH} shows that the entries of $\xvec{h}_E$ are positive for all $\rho \geq 0$ and $\alpha > 0$. Hence, the process is well-defined and there are no further restrictions needed, as in the case for the spARCH model.

\begin{corollary}\label{cor:EspARCH}
Assume that the assumptions of Theorem \ref{th:EspARCH} are fulfilled, then $\xvec{h}_{E}(\xvec{s}_i) \geq 0$ for all $i = 1, \ldots, n$.
\end{corollary}

For all proofs, we refer to the Appendix.

\subsection{Complex Spatial ARCH model}\label{sec:complexspARCH}

Now, we propose a complex-valued spARCH process. In order to obtain a solution of $\mathrm{diag}(\xvec{h})^{1/2}$ in the $n$-dimensional space of real numbers for the model defined in \eqref{eq:h_spARCH}, all elements of the matrix $(\xmat{I} - \xmat{A}^2)^{-1}$ must be nonnegative (see \citealt{Otto16_arxiv}). For the complex spARCH process, we relax the assumption that there should be a solution to $\mathrm{diag}(\xvec{h})^{1/2}$ in the real numbers and also consider  complex solutions.
Thus, the definition of $\xvec{h}$ coincides with $\xvec{h}_O$ of the original model, i.e.,
\begin{equation}\label{eq:complexspARCH}
h_{C}(\xvec{s}_i) = \alpha + \sum_{v=1}^{n} \rho w_{iv} Y(\xvec{s}_v)^2 \, .
\end{equation}

\subsection{Spatiotemporal ARCH model}\label{sec:spatiotemporal}

Finally, we show that spatiotemporal processes are covered directly by these approaches. For spatiotemporal data, the vector $\xvec{s}$ simply includes both the spatial location $\xvec{s}_s$ and the point in time $t$, i.e.,  $\xvec{s} = (\xvec{s}_s, t)'$. In addition, it is important to assume that future observations do not influence past observations, i.e., the weights $w_{ij}$ must be zero if $t_j \geq t_i$. However, the dimension of the weighting matrix $\xmat{W}$ might become very large for this representation. More precisely, the matrix has dimension $NT \times NT$, where $N$ is the total number of spatial locations and $T$ stands for the total number of time points. From a computational perspective, this is not necessarily a drawback since $\xmat{W}$ is usually sparse and could also have a block diagonal structure. Moreover, it is often reasonable to assume that $h(\xvec{s}_i)$ is only influenced by the neighbors of $\xvec{s}_{s, i}$ at the same point of time and by past observations at the same location. Then the weighting matrix would have the following structure
\begin{equation*}
\xmat{W} = \left(\begin{array}{cccc}
\xmat{W}_1    & \xmat{0}      & \cdots & \xmat{0} \\
\xmat{I}      & \xmat{W}_2    & \cdots & \xmat{0} \\
\vdots        & \vdots        & \ddots & \vdots        \\
\xmat{I}      & \xmat{I}      & \cdots & \xmat{W}_T    \\
\end{array}\right) \, .
\end{equation*}
Indeed, it is plausible to weight the spatial and temporal lags differently by replacing $\rho \xmat{W}$ by a sum
\begin{equation*}
\rho \left(\begin{array}{cccc}
\xmat{W}_1    & \xmat{0}      & \cdots & \xmat{0} \\
\xmat{0}      & \xmat{W}_2    & \cdots & \xmat{0} \\
\vdots        & \vdots        & \ddots & \vdots        \\
\xmat{0}      & \xmat{0}      & \cdots & \xmat{W}_T    \\
\end{array}\right)
+
\phi_1 \left(\begin{array}{cccc}
\xmat{0}      & \xmat{0}      & \cdots & \xmat{0} \\
\xmat{I}      & \xmat{0}      & \cdots & \xmat{0} \\
\vdots        & \vdots        & \ddots & \vdots        \\
\xmat{0}      & \xmat{0}      & \cdots & \xmat{0}      \\
\end{array}\right)
+ \ldots
\,
\end{equation*}
with positive weights $\phi_k$ for all temporal lags $1 \leq k \leq p$.

\begin{table}
\begin{center}
\begin{tabular}{
p{0.12\textwidth}>{\raggedright\arraybackslash}
p{0.4\textwidth}>{\raggedright\arraybackslash}
p{0.45\textwidth}}
\hline
Process type & Definition of $\xvec{h}$ & Comments \\ 
\hline
spARCH               & $\xvec{h}_O = \alpha\xvec{1} + \rho\xmat{W} \left(\xmat{I} - \xmat{A} \right)^{-1} (\alpha \xvec{\varepsilon}^{(2)})$ & $\xvec{\varepsilon}$ is simulated from multivariate normal distribution (MN) truncated on the interval $\left[-1/\sqrt[4]{||\rho^2\xmat{W}^2||_1}, 1/\sqrt[4]{||\rho^2\xmat{W}^2||_1}\right]$ \\  
spARCH (oriented)    & $\xvec{h}_O = \alpha\xvec{1} + \rho\xmat{W} \left(\xmat{I} - \xmat{A} \right)^{-1} (\alpha \xvec{\varepsilon}^{(2)})$ & $\xvec{\varepsilon} \sim \text{MN}(\xvec{0}, \xmat{I})$, $\xmat{W}$ must be a strictly triangular weighting matrix                                                                \\  
spatial E-ARCH       & $\ln \xvec{h}_E = \xmat{S} \left(\alpha\xvec{1} + \rho b \xmat{W} \ln|\xvec{Y}|\right)$ &  $\xvec{\varepsilon} \sim \text{MN}(\xvec{0}, \xmat{I})$, but moments of $\xvec{Y}$ differ from the moments of classical spARCH process (cf. \citealt{Otto16_arxiv})                                                           \\  
spARCH (complex)   & $\xvec{h}_C = \alpha\xvec{1} + \rho\xmat{W} \left(\xmat{I} - \xmat{A} \right)^{-1}(\alpha \xvec{\varepsilon}^{(2)})$ & $\xvec{\varepsilon} \sim \text{MN}(\xvec{0}, \xmat{I})$, but complex-valued $\xvec{Y}$    \\  
\hline
\end{tabular}
\end{center}
\caption{Overview of all types of spARCH models implemented in the \pkg{spGARCH} package.} \label{table:ARCHmodels}
\end{table}

\subsection{Spatial ARCH Disturbances}\label{sec:SARspARCH}

Since all conditional and unconditional odd moments of spatial ARCH processes are equal to zero, these ARCH-type models can easily be added to any kind of (spatial) regression model without influencing the mean equation as well as the spatial dependence in the first conditional and unconditional moments. This makes the spatial ARCH models flexible tools for dealing with conditional spatial heteroscedasticity in the residuals of spatial models. For instance, one can consider spatial autoregressive models for $\xvec{Y}$, i.e.,
\begin{eqnarray}
\xvec{Y} = \lambda \xmat{B} \xvec{Y} +  \xmat{X} \xvec{\beta} + \xvec{u} \label{eq:SARspARCH1}
\end{eqnarray}
with $\xvec{u}$ following either a spatial ARCH model with the original definition $\xvec{h}_O$ or the exponential model with $\xvec{h}_E$. Thus,
\begin{eqnarray}
\xvec{u} = \mathrm{diag}(\xvec{h})^{1/2}\xvec{\varepsilon} \, . \label{eq:SARspARCH2}
\end{eqnarray}
Further, we call this model the SARspARCH model. For $\lambda = 0$, the model collapses to a simple linear regression model; if, additionally, $\xvec{\beta} = \xvec{0}$, the model coincides with the previously discussed ARCH models. Thus, these coefficients can be used for testing against nested models.

In contrast to other models for heteroscedastic errors, such as the SAR\-AR or SAR\-MA models, which assume spatial autoregressive or spatial moving average error terms (cf. \citealt{Kelejian10,Fingleton08b,Haining78}), the SARspARCH model does not affect the spatial autocorrelation of the process, just the spatial heteroscedasticity, because all conditional and unconditional odd moments are equal to zero. Thus, $\lambda \xmat{B}$ can be interpreted directly as the spatial dependence of the process, while $\rho \xmat{W}$ describes the spatial dependence in the second conditional moments. Moreover, these two parts can be interpreted separately, as we will demonstrate in the last section via an empirical example.

\section{Parameter Estimation}\label{sec:estimation}

The parameters of a spatial ARCH process can be estimated by the maximum-likelihood approach. To obtain the joint density for $\xvec{Y} = k(\xvec{\varepsilon})$, the Jacobian matrix of $k^{-1}$ at the observed values $\xvec{y}$ must be computed (e.g., \citealt{Bickel15}). If $f_{\xvec{\varepsilon}}$ is the distribution of the error process, then the joint density $f_{\xvec{Y}}$ of $\xvec{Y}$ is given by
\begin{eqnarray}
f_{\xvec{Y}}(\xvec{y}) & =  & f_{(Y(\xvec{s}_1), \ldots, Y(\xvec{s}_n))}(y_1, \ldots, y_n) \nonumber \\
& = & f_{\xvec{\varepsilon}}\left(\frac{y_1}{\sqrt{h}_1}, \ldots, \frac{y_n}{\sqrt{h}_n}\right) | \det\left( \left( \frac{\partial y_j/\sqrt{h_j}}{\partial y_i} \right)_{i,j=1,\ldots,n}\right) | \, . \label{eq:transformation}
\end{eqnarray}
If the residuals are additionally independent and identically distributed, the parameter estimates can be obtained from the maximization of the log-likelihood as follows
\begin{eqnarray*}
(\hat{\alpha}, \hat{\rho}) = \underset{\alpha > 0, \rho \geq 0}{\arg\max} \; \ln | \det\left( \left( \frac{\partial y_j/\sqrt{h_j}}{\partial y_i} \right)_{i,j=1,\ldots,n}\right) | +  \sum_{i = 1}^{n} \ln f_{\varepsilon}(y_i) \, .
\end{eqnarray*}

The Jacobian matrix, of course, depends on the definition of $\xvec{h}$. For the spARCH process, this Jacobian matrix can be specified as
\begin{equation*}
\frac{\partial y_j/\sqrt{h_j}}{\partial y_i} = \left\{
\begin{array}{ccc}
1\, /\, \sqrt{h}_j & \mbox{for} & i=j \\
- \frac{y_i y_j}{h_j^{3/2}} \rho w_{ji} & \mbox{for} & i \neq j
\end{array} \right. \, .
\end{equation*}
In contrast, the Jacobian matrix for the E-spARCH process is slightly different, namely
\begin{equation*}
\frac{\partial y_j/\sqrt{h_j}}{\partial y_i} = \left\{
\begin{array}{ccc}
1\, /\, \sqrt{h_j} & \mbox{for} & i = j \\
- \frac{b y_j}{2 y_i h_j^{3/2}} \rho s_{ji} w_{ji} & \mbox{for} & i \neq j
\end{array} \right. \,
\end{equation*}
with
\begin{equation*}
h_j = \mathrm{exp}\left(\sum_{v = 1}^{n} s_{jv}\left(\alpha +\rho w_{jv} \ln |y_v| \right) \right) \, .
\end{equation*}

From a computational perspective, the computation of the log determinant of this matrix is feasible for large data sets. To be precise, the log-determinant is equal to
\begin{equation*}
\ln | \det\left( \text{diag}\left(\frac{h_1}{y_1^2},\ldots, \frac{h_n}{y_n^2}\right) - \rho \xmat{W}^\prime \right) | + \sum_{i=1}^n \ln \frac{y_i^2}{h_i^{3/2}}
\end{equation*}
for the spARCH process. Similarly, it is given by
\begin{equation*}
\ln | \det\left( \mathrm{diag}\left(\frac{2 h_1}{b},\ldots, \frac{2 h_n}{b}\right) - \rho \xmat{S}^\prime \circ \xmat{W}^\prime \right) | +  \sum_{i=1}^n \ln \frac{b}{2 h_i^{3/2}} \, .
\end{equation*}
for the E-spARCH process, where $\circ$ stands for the Hadamard product.

In the \pkg{spGARCH} package, we implemented the iterative maximization algorithm with inequality constraints proposed by \cite{Ye88}, which is implemented in the \proglang{R}-package \pkg{Rsolnp} (see \citealt{RSolnp}). It is important to note that the log determinant of the Jacobian also depends on the parameters in such a way that it needs to be computed in each iteration (see, also, Theorem 13.7.3 of \cite{Harville97} for the computation of a determinant for the sum of a diagonal matrix and an arbitrary matrix), but $\xmat{W}$, and therefore $\xmat{S} \circ \xmat{W}$, are usually sparse. Thus, the required time for the estimation of the parameters depends mainly on the dimension and sparsity of $\xmat{W}$.

\section[Overview of the R-Package spGARCH]{Overview of the \proglang{R}-Package \pkg{spGARCH}}\label{sec:spGARCH_package}

The \proglang{R}-package \pkg{spGARCH} provides several basic functions for the analysis of spatial data showing spatial conditional heteroscedasticity. In particular, the process can be simulated for arbitrarily chosen weighting matrices according to the definitions in Section \ref{sec:models}. Moreover, we implement a function for the computation of the maximum-likelihood estimators. To generate a user-friendly output, the object generated by the estimation function can easily be summarized by the generic \fct{summary} function. We also provide all common generic methods, such as \fct{plot}, \fct{print}, \fct{logLik}, and so forth. To maximize the computational efficiency, the actual version of the package contains compiled \proglang{C++} code (using the packages \pkg{Rcpp} and \pkg{RcppEigen}, cf. \citealt{Rcpp,RcppEigen}). A brief overview of the package and its main functions is given in Table \ref{table:functions}. Further, we focus on the two main aspects of the package, i.e., the simulation (described in detail in Section \ref{sec:comp_simulation}) and estimation (Section \ref{sec:comp_estimation}) aspects of the spARCH, E-spARCH, and SARspARCH processes.

\begin{table}
\begin{center}
\begin{tabular}{
p{0.3\textwidth}
p{0.65\textwidth}}
\hline
Function & Description \\
\hline
\emph{Main functions} & \\
$\quad$ \fct{sim.spARCH}     & Simulation of spARCH and E-spARCH processes\\
$\quad$ \fct{qml.spARCH}     & Quasi-maximum-likelihood estimation for spARCH models \\
$\quad$ \fct{qml.SARspARCH}  & Quasi-maximum-likelihood estimation for SAR models with spARCH residuals \\
\emph{Generic methods} & \\
$\quad$ \fct{summary}        & Summary of an object of \class{spARCH} class generated by \fct{qml.spARCH} or \fct{qml.SARspARCH} \\
$\quad$ \fct{print}          & Printing method for \class{spARCH} class or \code{summary.spARCH} class \\
$\quad$ \fct{fitted}         & Extracts the fitted values of an object of  \class{spARCH} class \\
$\quad$ \fct{residuals}      & Extracts the residuals of an object of  \class{spARCH} class \\
$\quad$ \fct{logLik}         & Extracts the log-likelihood of an object of  \class{spARCH} class \\
$\quad$ \fct{extractAIC}     & Extracts the AIC of an object of  \class{spARCH} class \\
$\quad$ \fct{plot}           & Provides several descriptive plots of the residuals of an object of  \class{spARCH} class \\
\hline
\end{tabular}
\end{center}
\caption{Summary of the main functions of the \pkg{spGARCH} package.} \label{table:functions}
\end{table}

\subsection{Simulation of ARCH-type stochastic processes}\label{sec:comp_simulation}

The simulations of all spatial ARCH-type models are implemented in one function, namely, the \fct{sim.spARCH} function. The different definitions of the model are specified via the argument \code{type}. The use of \fct{sim.spARCH} is very similar to how a basic random number generator is used, meaning that the first argument \code{n} is the number of generated values and all further arguments specify the parameters of the spARCH process. For instance, one might simulate an oriented spARCH process (meaning $\xmat{W}$ is triangular) on a $d \times d$ spatial lattice with $\rho = 0.7$ and $\alpha = 1$ using the following lines.
\begin{lstlisting}
R> require("spdep")
R> rho              <- 0.7
R> alpha            <- 1
R> d                <- 50
R> n                <- d^2
R> nblist           <- cell2nb(d, d, type = "queen")
R> W                <- nb2mat(nblist)
R> W[upper.tri(W)]  <- 0
R> Y                <- sim.spARCH(n = n, rho = rho, alpha = alpha, W = W, type = "gaussian", control = list(seed = 5515))
\end{lstlisting}
To build the spatial weighting matrix, we used \fct{cell2nb} from the \pkg{spdep} package, returning an \code{nb} object of a $d \times d$ lattice (see \citealt{Cressie93,spdep}). Further, we converted the \code{nb} object into a contiguity matrix, as \fct{sim.spARCH} requires either a matrix (class \code{matrix}) or a sparse matrix (class \code{dgCMatrix}) as an argument. Usually, spatial weighting matrices are sparse by construction. Thus, $\xmat{W}$ is always converted internally to a \code{dgCMatrix} matrix or rather to a \code{SparseMatrix} object defined in the \pkg{eigen} library in \proglang{C++}.
Via the \code{control} parameter, a random seed might be passed to the simulation function. If not, a random seed is assigned randomly from a uniform distribution and printed in console in order that one might reproduce the result even without having a random seed specified in advance. We prefer to print a single number in the console rather than returning to the random number generator (RNG) state as an attribute of the returned vector. Thus, a random seed might either be passed as an optional argument to \fct{sim.spARCH} or set before calling \fct{sim.spARCH} by \fct{set.seed}.

There are several types of spatial ARCH processes which can be simulated by \fct{sim.spARCH}. They are all specified by the argument \code{type}. If
\begin{itemize}
\item \code{type = "gaussian"}, then the original spARCH process according to the definition in \cite{Otto16_arxiv} is simulated.
\begin{itemize}
\item If there exists a permutation such that $\xmat{W}$ is a strictly triangular matrix, then the function simulates automatically an oriented spARCH process with independent and identically gaussian distributed errors.
\item If there is no such permutation, then the errors are simulated from a truncated normal distribution with $a = 1 / \sqrt[4]{\rho^2 ||\xmat{W}||_1}$.
\end{itemize}
\item \code{type = "exp"}, an E-spARCH process is simulated with an user-specified value of $b$ (default 2) and standard normal random errors.
\item \code{type = "complex"}, complex solutions of $\mathrm{diag}(\xvec{h})^{1/2}$ are considered in order to simulate the spARCH process.
\end{itemize}

Figure \ref{fig:simulations} illustrates the behavior of different types of spatial ARCH processes. All of them are simulated with the same parameters and random seeds in such a manner that the vector $\xvec{\varepsilon}$ is identical for all types of processes, except for the spARCH process with the truncated normal errors. In the first row, the spatial weighting is achieved via a strictly triangular Queen's contiguity matrix, which means that the spatial dependence has its origin in the upper left corner. To the contrary, $\xmat{W}$ presents a classical Queen's contiguity matrix in the second row. We additionally plot a spatial white noise process for comparison, as we used a rather unconventional two-color scheme. Using this kind of color scheme, one might distinguish between positive and negative observations, such that it is easier to see the spatial volatility clusters. Areas of smaller volatility are characterized by rather evenly gray pixels, whereas clusters of high volatility have rather intense colors. Moreover, the colors fluctuate irregularly between blue and red.

\begin{figure}
  \begin{center}
  \includegraphics[width=0.32\textwidth]{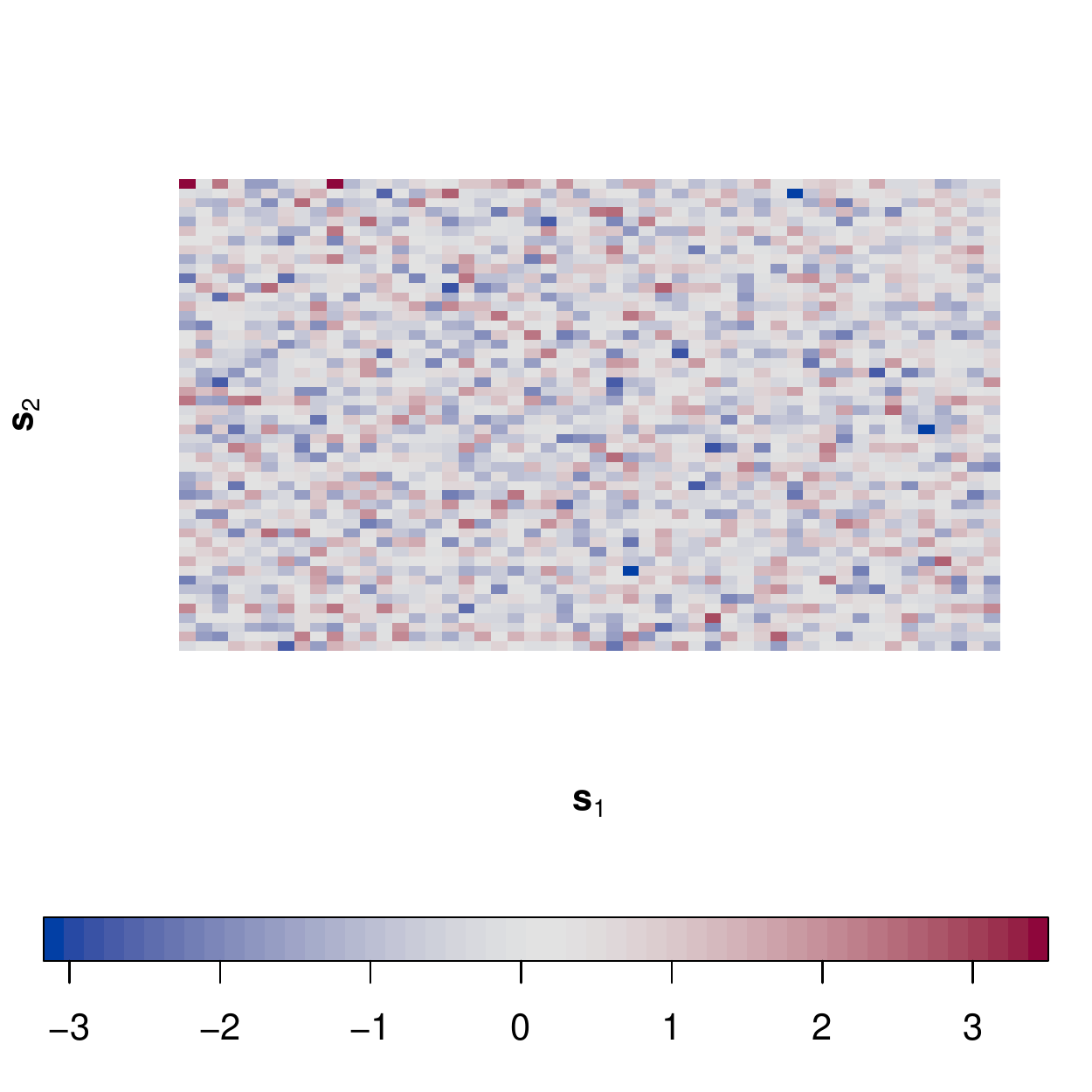}
  \includegraphics[width=0.32\textwidth]{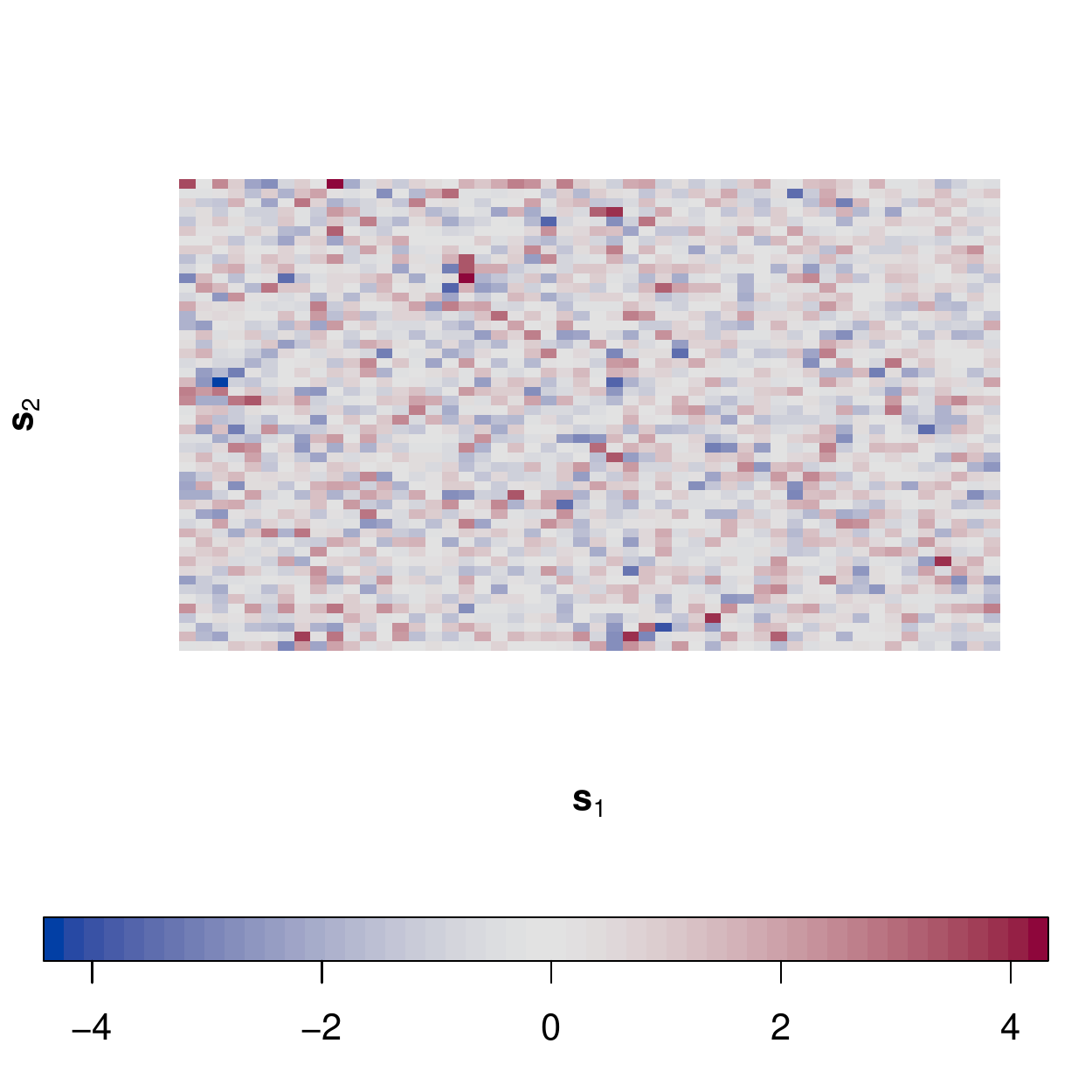}
  \includegraphics[width=0.32\textwidth]{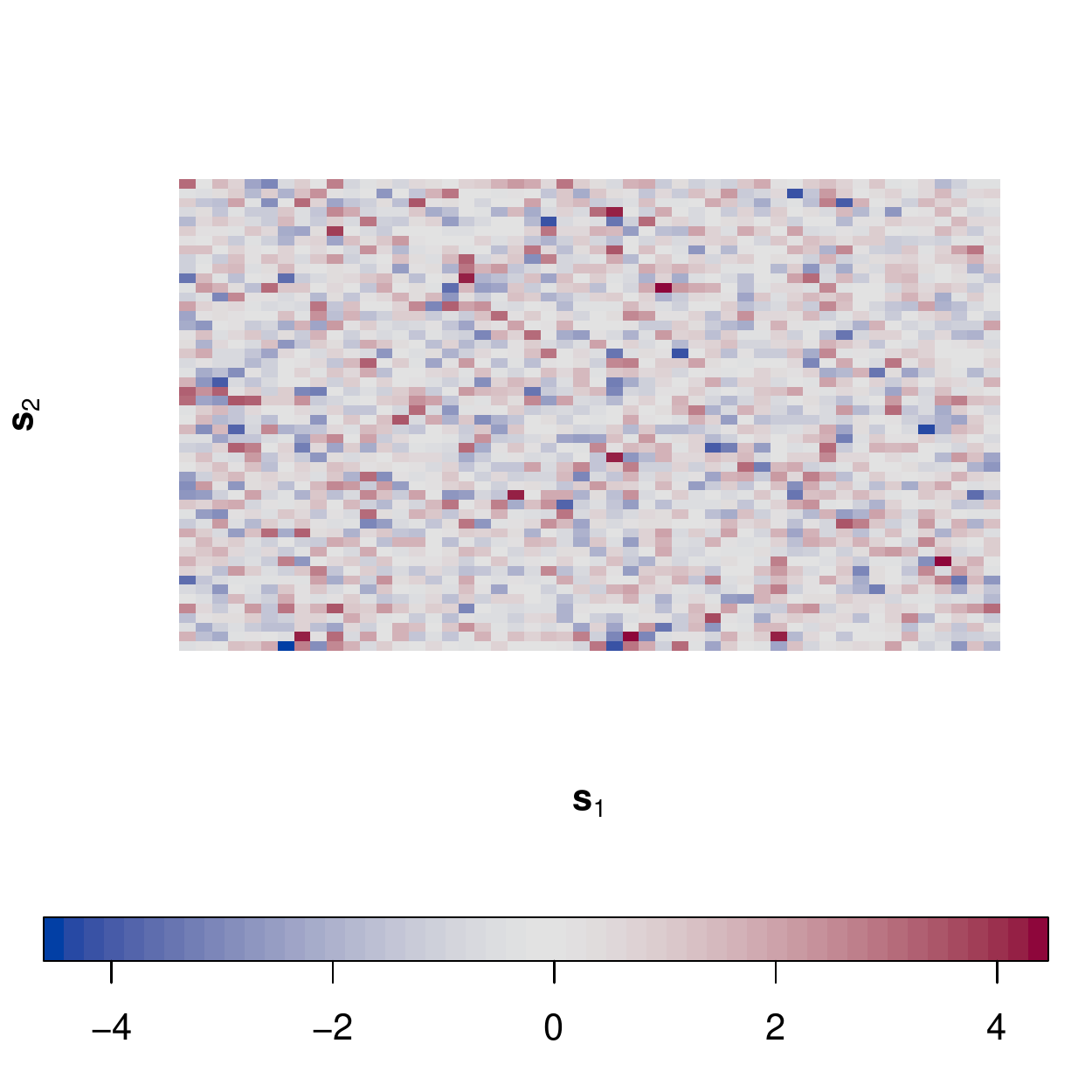}\\
  \includegraphics[width=0.32\textwidth]{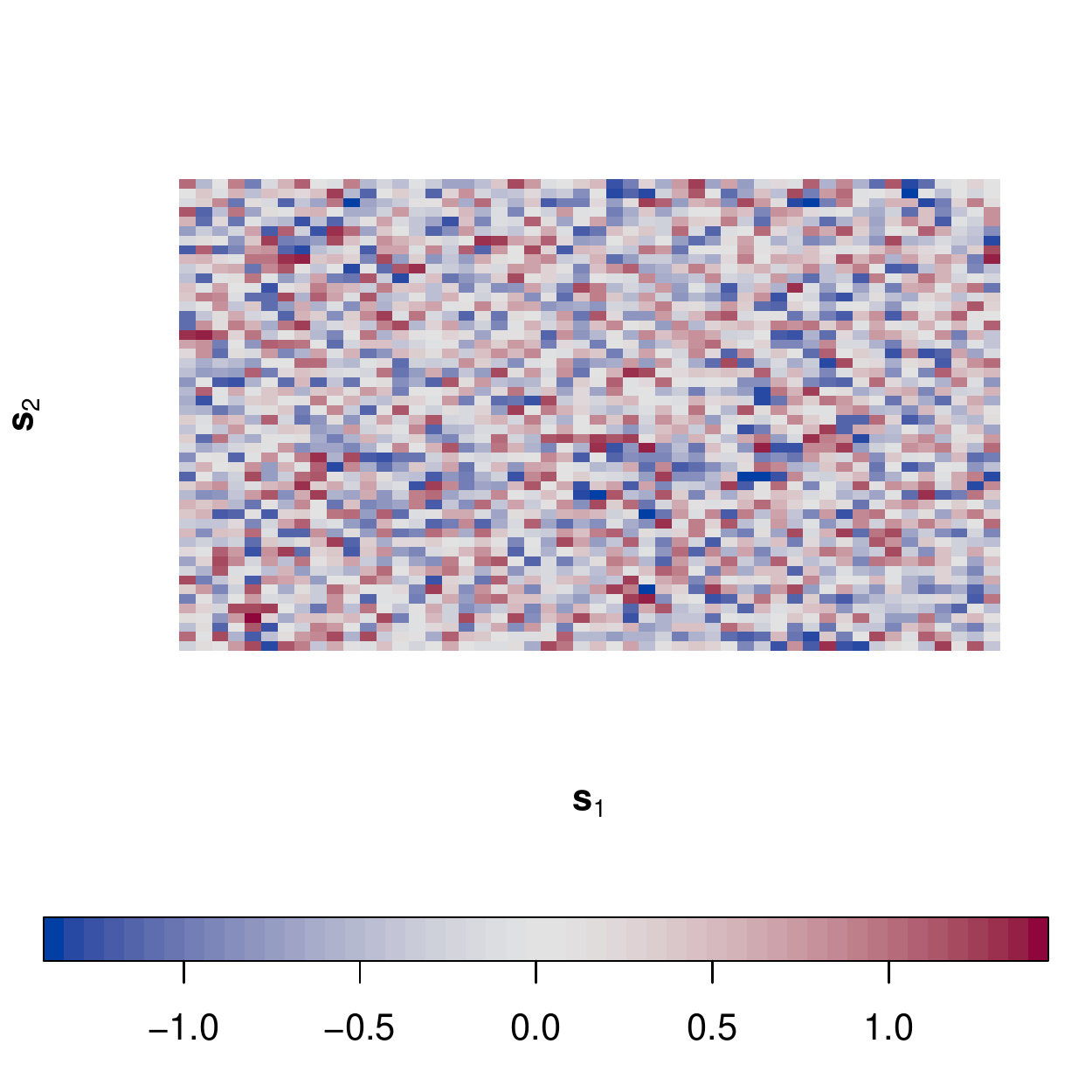}
  \includegraphics[width=0.32\textwidth]{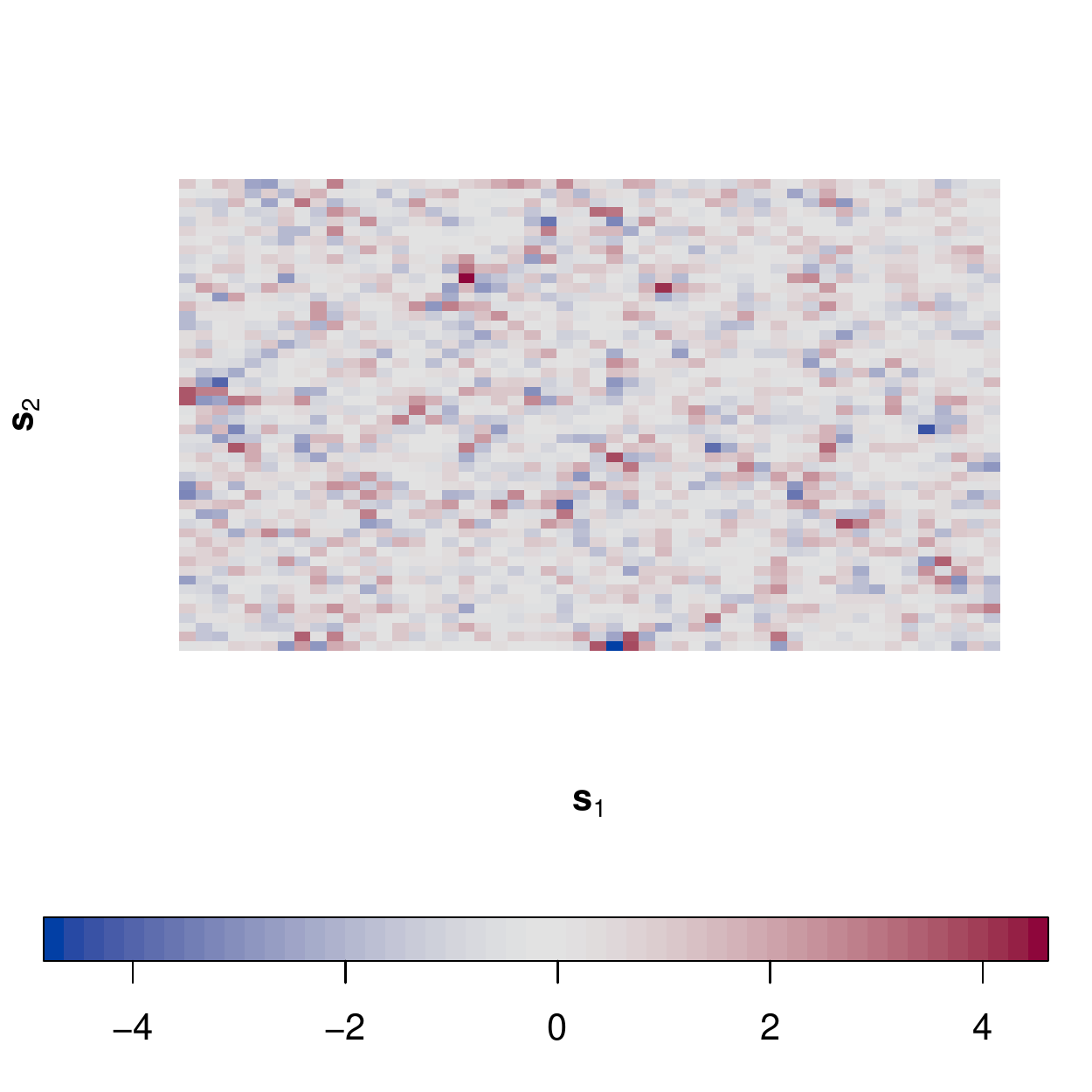}
  \includegraphics[width=0.32\textwidth]{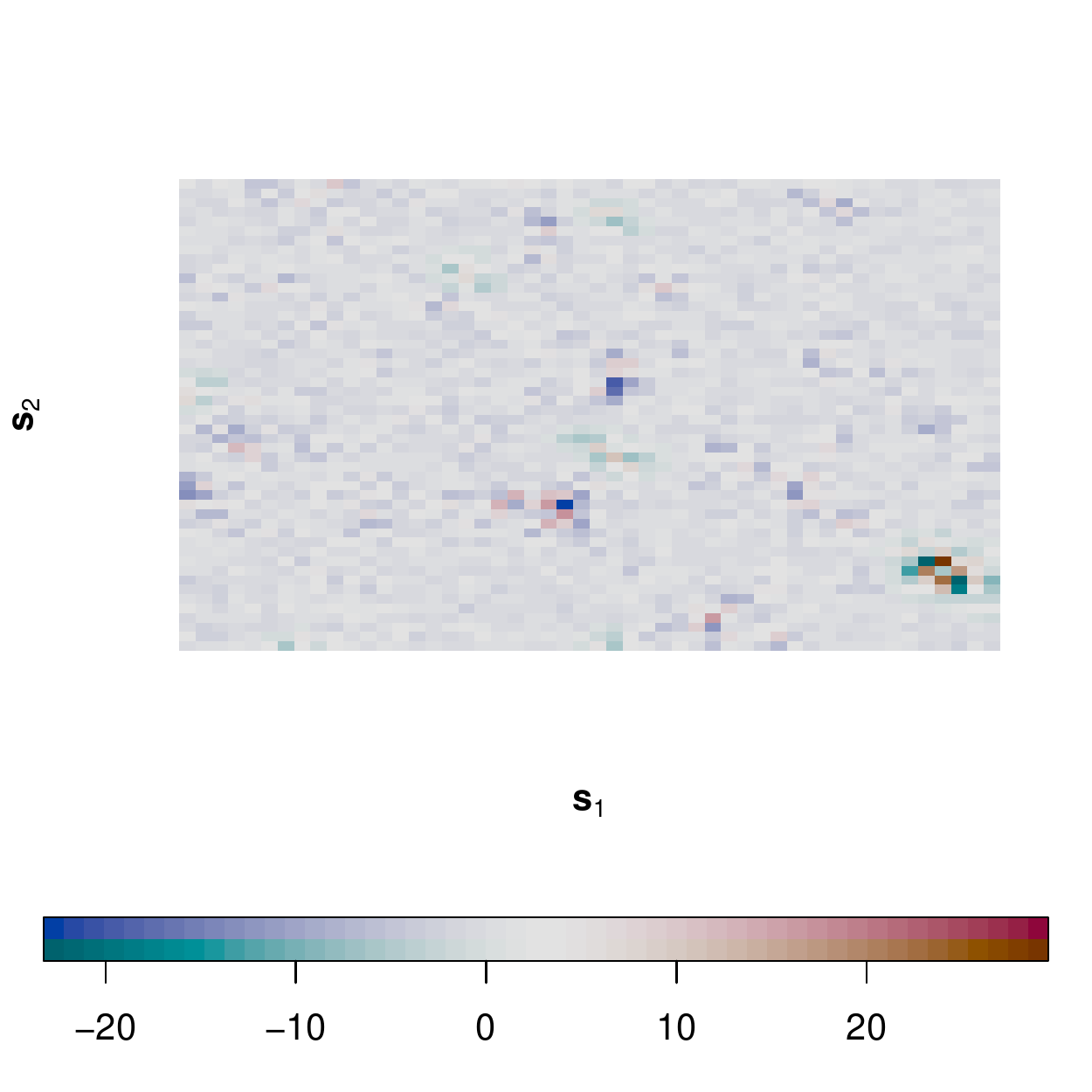}
  \end{center}
  \begin{footnotesize}
  Above left: spatial white noise for comparison; center: oriented spARCH (\code{type = ``gaussian''}); right: spatial E-ARCH (\code{type = ``exp''}).\\[.1cm]
   Below left: spARCH with truncated normal errors (\code{type = ``gaussian''}); center: spatial E-ARCH (\code{type = ``exp''}), right: complex spARCH (\code{type = ``complex''}).
  \end{footnotesize}
  \caption{Simulations on a two-dimensional lattice for triangular matrices (above) and non-triangular matrices (below). For all simulations, we set $\rho = 0.7$ and $\alpha = 1$, and $\xmat{W}$ is chosen to be the Queen's contiguity matrix.}\label{fig:simulations}
\end{figure}

\subsection{Maximum-likelihood estimation}\label{sec:comp_estimation}

Other important functions of the package are the \fct{qml.spARCH} and \fct{qml.SARspARCH} functions, which implement a quasi-maximum-likelihood estimation algorithm (QML). As for the \fct{sim.spARCH} function, many spARCH models are covered in the \fct{qml.spARCH} and \linebreak \fct{qml.SARspARCH} function. Thus, the user needs to specify which particular spARCH model is to be fitted via the argument \code{type}. Moreover, the model for the mean equation is a user-specified \code{formula}, making the use of the estimation functions similar to the use of the common \fct{lm} or \fct{glm} functions.

In general, the estimators exhibited good performances for a variety of error distributions in simulation studies, although the likelihood function was derived under the normality assumption. This is not surprising, as the maximum-likelihood estimators have good properties under mild assumptions for the error processes of a variety of similar spatial econometrics models (cf. \citealt{Lee04,Lee12,Lee10d,Lee10b}). Thus, we refer to the approach as the QML approach, and the name of the estimation functions start with \code{qml} instead of \code{ml}. In the following paragraphs, we start the simulation of one specific sample, which is then used further to illustrate the log-likelihood functions as well as to demonstrate parameter estimation.

Compared to the E-spARCH processes, the likelihood functions of spARCH models are rather flat around the global maximum. This behavior is illustrated for simulated processes in Figure \ref{fig:LL}. The observations for the E-spARCH process have been simulated as follows.
\begin{lstlisting}
R> nblist          <- cell2nb(20, 20, type = "queen")
R> W               <- nb2mat(nblist)
R> y               <- sim.spARCH(n = 20^2, rho = 0.5, alpha = 1, W = W, type = "exp", control = list(seed = 5515))
\end{lstlisting}
To simulate an oriented process, the entries of $\xmat{W}$ above the diagonal must be set to zero and the argument \code{type} must be changed to \code{"gaussian"}, i.e.,
\begin{lstlisting}
R> W[upper.tri(W)] <- 0
R> y2              <- sim.spARCH(n = 20^2, rho = 0.5, alpha = 1, W = W, type = "gaussian", control = list(seed = 5515))
\end{lstlisting}

\begin{figure}
  \centering
  \textbf{spARCH process}\\
  \includegraphics[width=0.32\textwidth]{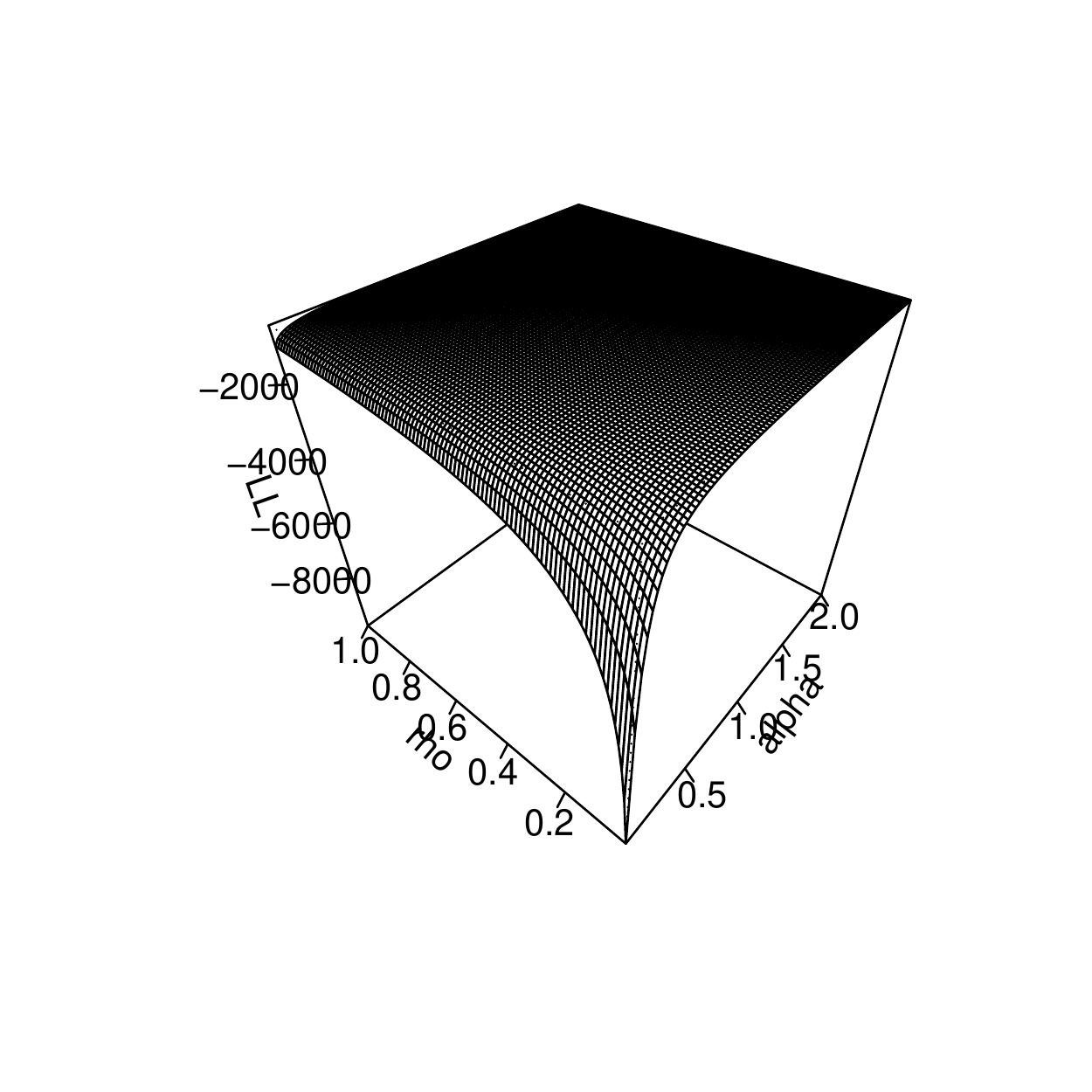}
  \includegraphics[width=0.32\textwidth]{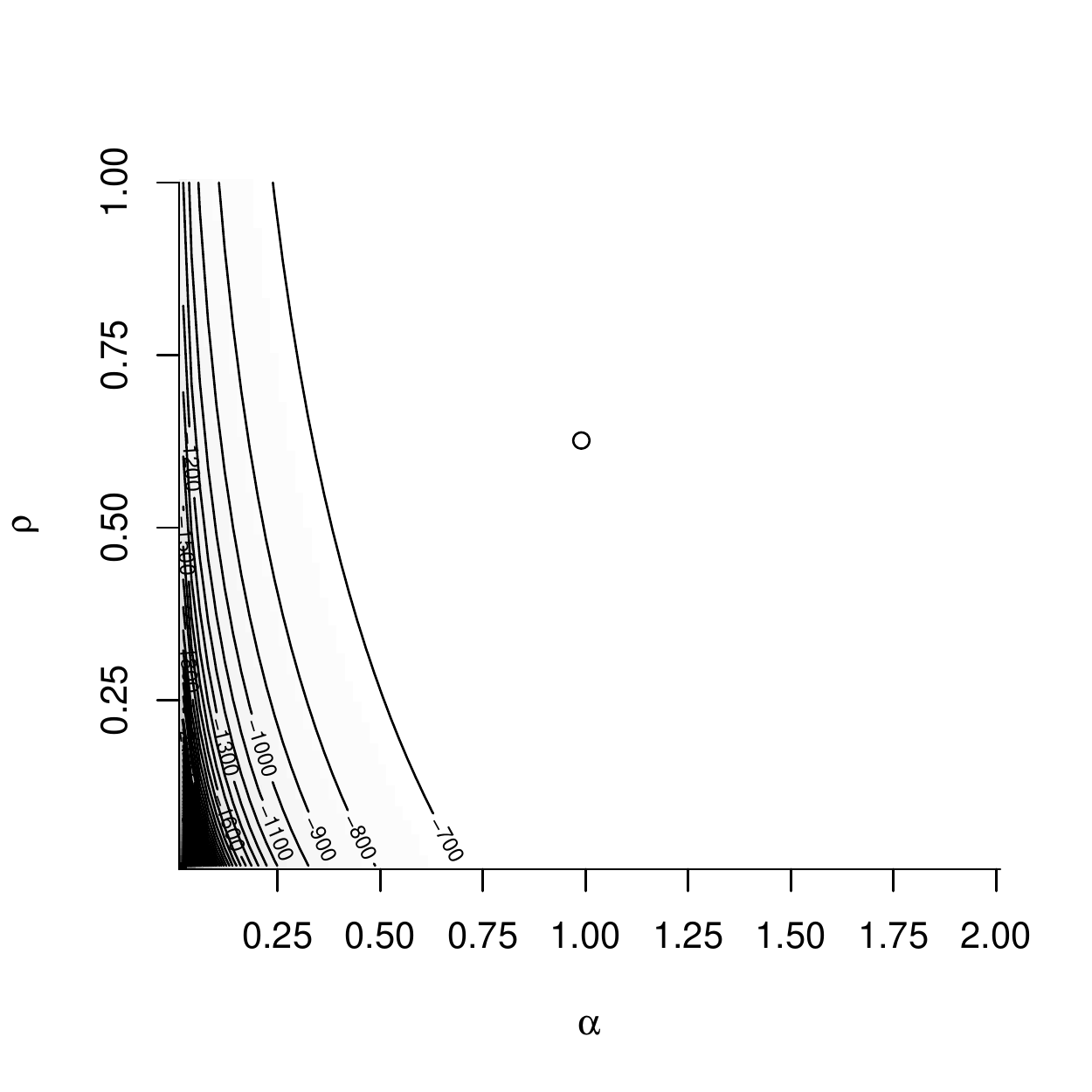}
  \includegraphics[width=0.32\textwidth]{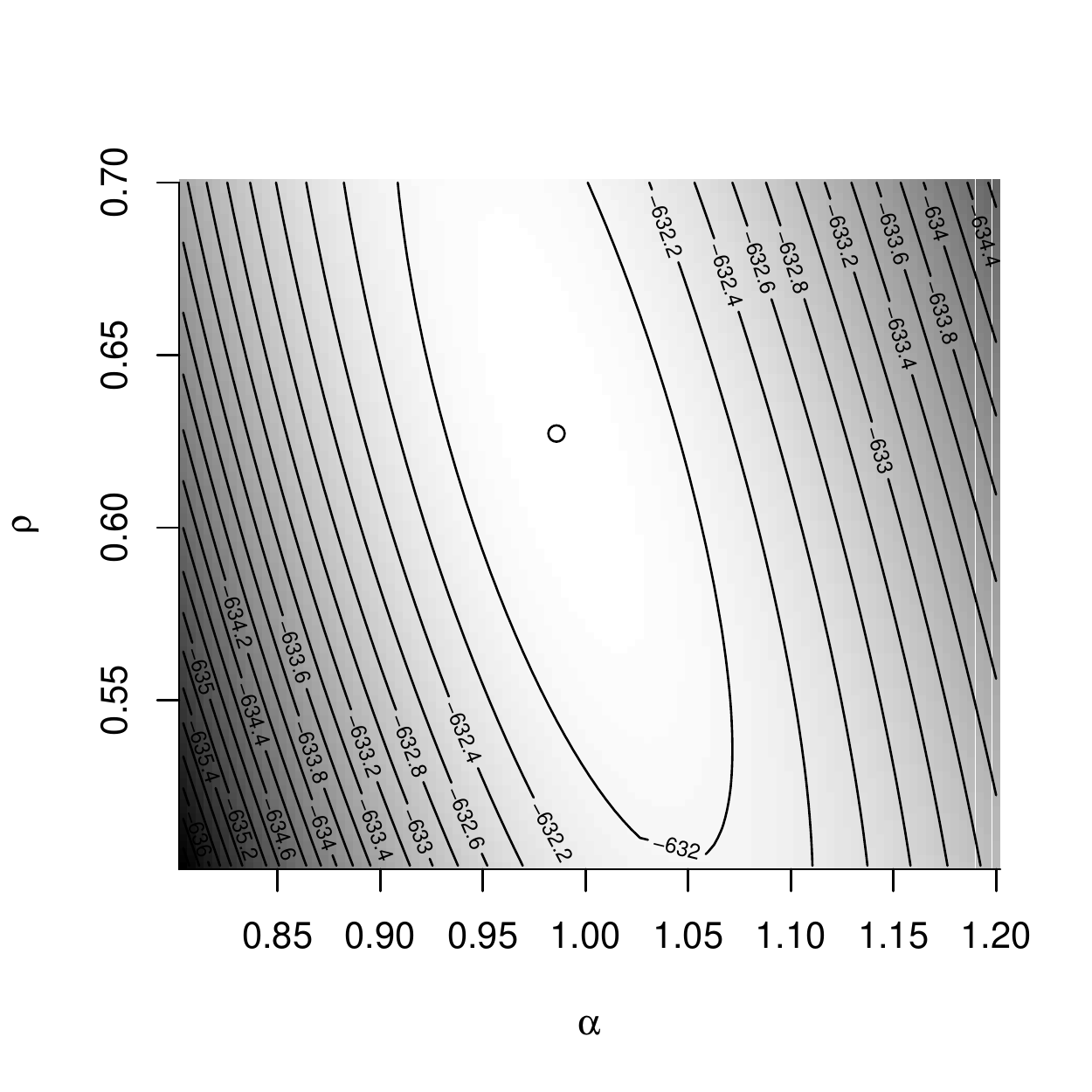}\\
  \textbf{E-spARCH process}\\
  \includegraphics[width=0.32\textwidth]{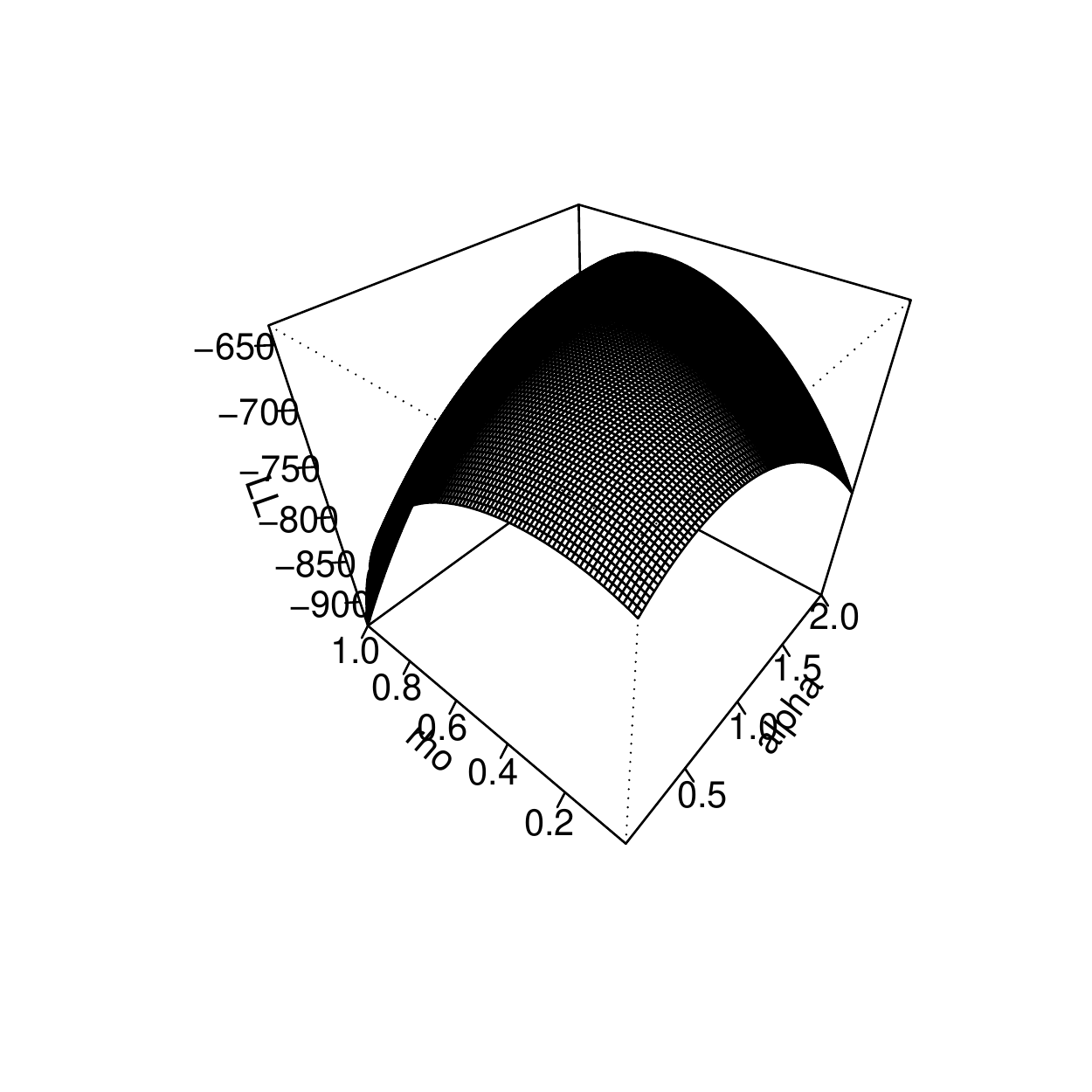}
  \includegraphics[width=0.32\textwidth]{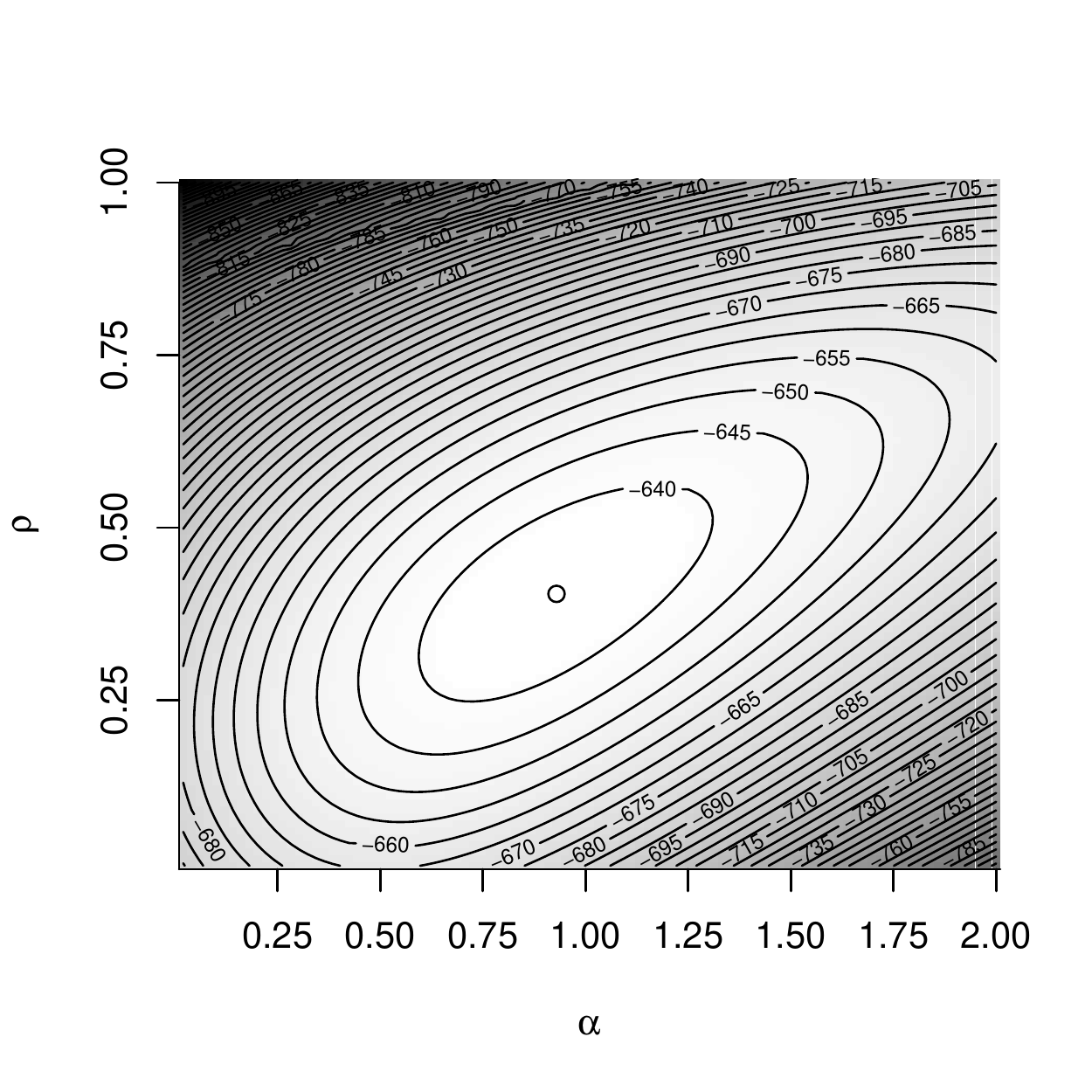}
  \includegraphics[width=0.32\textwidth]{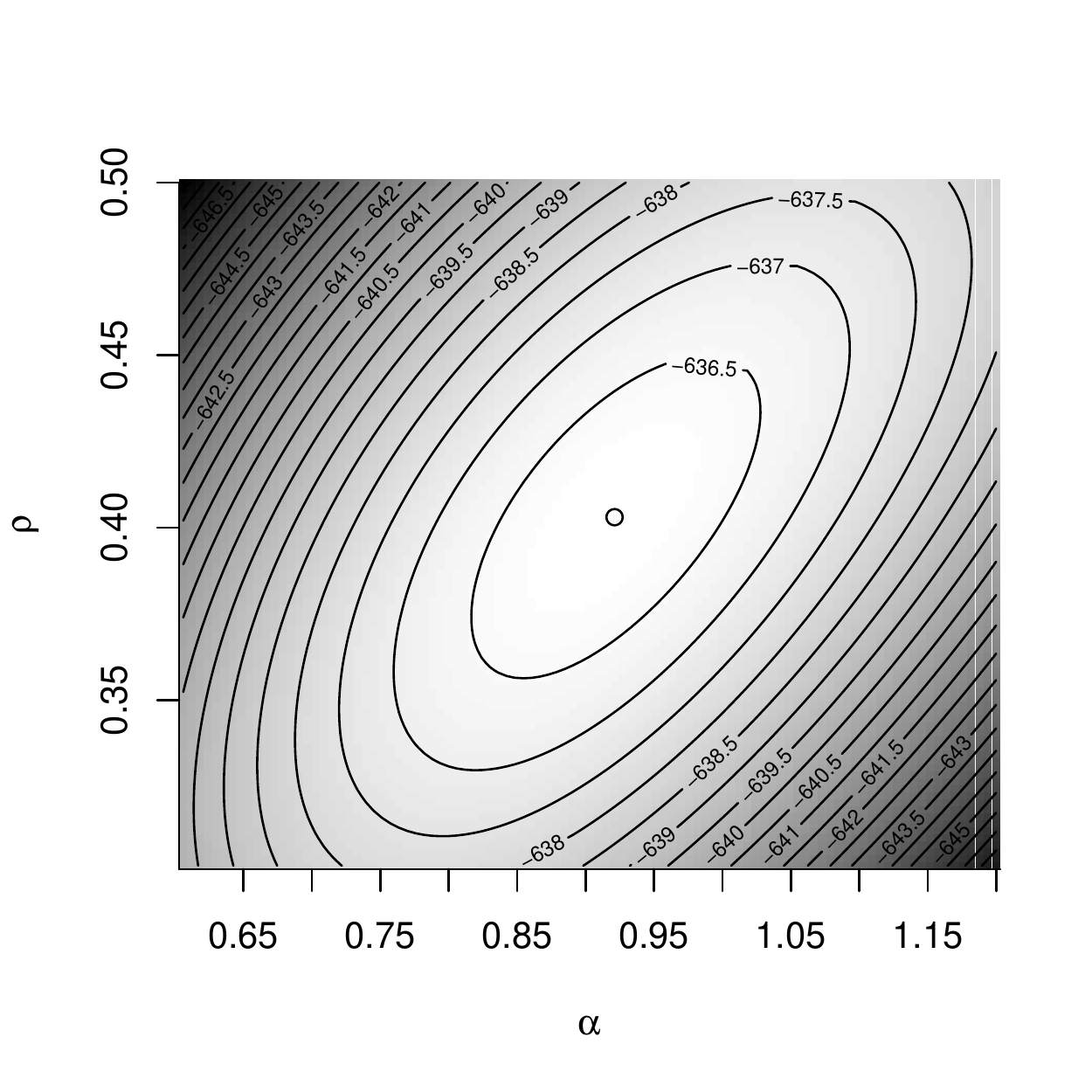}
  \caption{Logarithmic likelihood function.}\label{fig:LL}
\end{figure}

To estimate the parameters of an intercept-free E-spARCH model without any regressors, the formula passed to the function \fct{qml.spARCH} should be specified as \code{y $\sim$ 0}. In addition, a \code{data.frame} can be passed via the \code{data} argument to the \code{qml} functions. Although the likelihood function of a spARCH process is flat, good estimates can be obtained through iterative maximization. \cite{Otto16_arxiv} analyze the performance of the estimators in detail. The algorithm implemented in the packages is based on the \pkg{Rsolnp} package, allowing for both equality and inequality parameter constraints (cf. \citealt{RSolnp}).

The results of the estimation procedure are returned via an object of the class \class{spARCH}, for which we provide additionally several generic functions. First, there is a \fct{summary} function for the \class{spARCH} object. The summary shows all important estimation results, i.e., the parameter estimates, standard errors, test statistics, and asymptotic p-values, including significance stars. The estimation of the above simulated E-spARCH process would return the following results.
\begin{lstlisting}
R> spARCH_object   <- qml.spARCH(y ~ 0, W = W, type = "exp")
R> summary(spARCH_object)
 Call:
qml.spARCH(formula = y ~ 0, W = W, type = "exp")

 Residuals:
      Min.    1st Qu.     Median       Mean    3rd Qu.       Max.
-2.6867629 -0.6197315 -0.0053580 -0.0002615  0.5708346  2.8576621

 Coefficients:
      Estimate Std. Error t value  Pr(>|t|)
alpha 0.919324   0.128544  7.1518 8.564e-13 ***
rho   0.402998   0.056519  7.1304 1.001e-12 ***
---
Signif. codes:  0 *** 0.001 ** 0.01 * 0.05 . 0.1   1

 AIC: 543.01, BIC: 539.01 (Log-Likelihood: -269.51)

 Moran's I (residuals): -0.028568, p-value: 0.31795

 Moran's I (squared residuals): 0.035239, p-value: 0.14479
\end{lstlisting}
The standard errors are estimated as Cramer-Rao bounds from the Hessian matrix of the log-likelihood function. For triangular weighting matrices, the estimators are asymptotically normally distributed (\citealt{Otto16_arxiv}). In addition to the Akaike and Bayesian Schwarz information criteria, the results of Moran's test on the residuals and squared residuals are reported for the spatial autocorrelation of the residuals. However, it is possible to use functions like \fct{AIC} or \fct{BIC}, since there is a \fct{logLik} method for the objects from class \class{spARCH}. Additionally, the fitted values and residuals can be extracted by \fct{fitted} and \fct{residuals}, respectively.

To analyze the residuals, we provide additionally several descriptive plots via the generic \fct{plot} function. The first two plots are produced by \fct{moran.plot} imported from the package \pkg{spdep}. They inspect the spatial autocorrelation of the residuals and the squared residuals. In addition, the error distribution is depicted in the third graphic by a normal Q-Q-plot. The output obtained for the above numerical example is given below and in Figure \ref{fig:plot_out}.
\newpage
\begin{lstlisting}
R> AIC(spARCH_object)
[1] 543.0126
R> BIC(spARCH_object)
[1] 550.9956
R> par(mfcol = c(1,3))
R> plot(spARCH_object)
Reproduce the results as follows:
	 eps <- residuals(x)
	 W <- as.matrix(x$W)
	 moran.plot(eps, mat2listw(W), zero.policy = TRUE,
     xlab = "Residuals", ylab = "Spatially Lagged Residuals")
Reproduce the results as follows:
	 eps <- residuals(x)
	 W <- as.matrix(x$W)
	 moran.plot(eps, mat2listw(W), zero.policy = TRUE,
     xlab = "Residuals", ylab = "Spatially Lagged Residuals")
Reproduce the results as follows:
	 eps <- residuals(x)
	 std_eps <- (eps - mean(eps))/sd(eps)
	 qqnorm(eps, ylab = "Standardized Residuals")
	 qqline(eps)
\end{lstlisting}

\begin{figure}
  \centering
  \includegraphics[width=\textwidth]{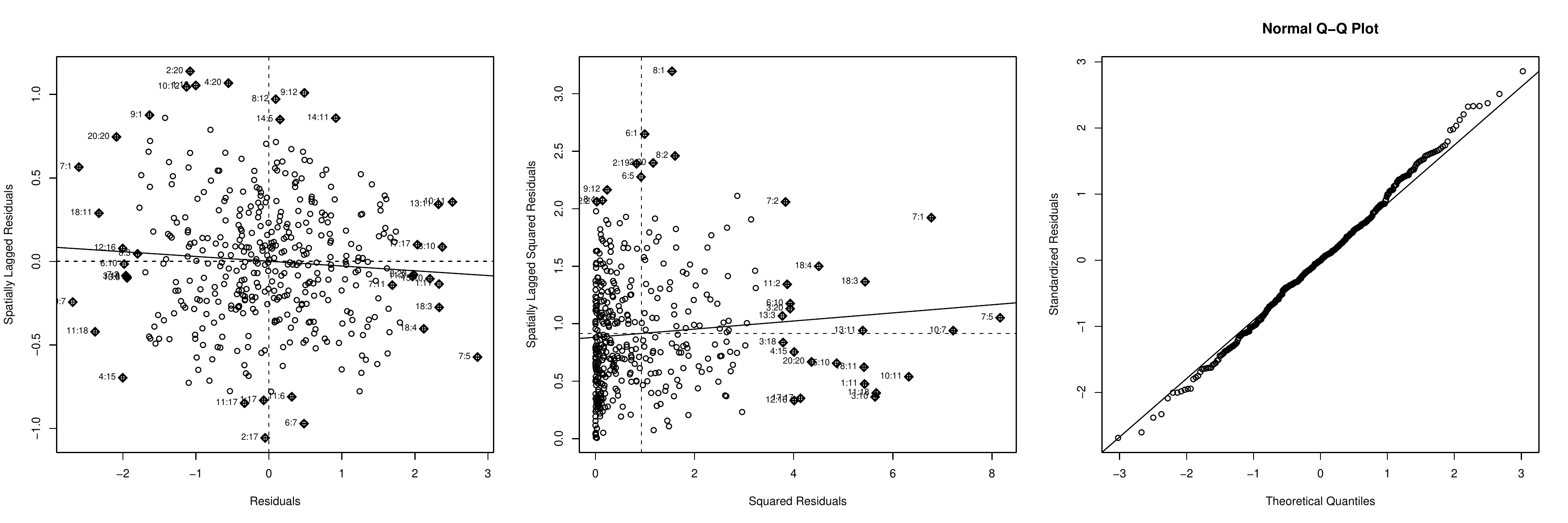}
  \caption{Resulting graphical output of \fct{plot}.}\label{fig:plot_out}
\end{figure}

The mean equation can be specified as \code{formula} for all models, i.e., the spARCH, E-spARCH, and SARspARCH models. Thus, there is a huge variety of possible spatial ARCH models as well as regression models with spARCH residuals which can be fitted by the estimation functions. In addition to linear models of the form \code{y $\sim$ a + b}, more sophisticated models can also be fitted, e.g., models with interactions \code{y $\sim$ a + b:c}, factor models \code{y $\sim$ factor}, polynomial models \code{y $\sim$ poly(a, 3)}, seasonally or regularly varying models of the form \code{y $\sim$ sin(t) + cos(t)} or \code{y $\sim$ sin(long) + cos(long) + sin(lat) + cos(lat)}, and so forth. We also implement an \fct{extractAIC} method for \class{spARCH} objects, such that one might also use \fct{step} for stepwise model selection. Table \ref{table:fittable_models} provides an overview of possible combinations of the arguments \code{formula} and \code{type} and shows the resulting models, which can be fitted by the functions \fct{qml.spARCH} and \fct{qml.SARspARCH}, respectively.

\begin{table}
  \centering
  \begin{tabular}{
  p{0.18\textwidth}
  p{0.145\textwidth}
  p{0.12\textwidth}
  p{0.50\textwidth}}
  \hline
  Function & \code{formula} & \code{type} & Resulting model \\
  \hline
  \fct{qml.spARCH}    & \code{y $\sim$ 0}        & \code{"gaussian"} & spARCH model (see \eqref{eq:spARCHy} and \eqref{eq:h_spARCH})                                                  \\
  \fct{qml.spARCH}    & \code{y $\sim$ 1}        & \code{"gaussian"} & spARCH model with an additional intercept for the mean equation                                                                     \\
  \fct{qml.spARCH}    & \code{y $\sim$ a + b}    & \code{"gaussian"} & Linear Regression with regressors \code{a} and \code{b} and spARCH residuals         \\
  \fct{qml.spARCH}    & \code{y $\sim$ a + b:c}  & \code{"gaussian"} & Linear Regression with more complex expressions and spARCH residuals          \\
  \fct{qml.spARCH}    & \code{y $\sim$ 0}        & \code{"exp"}      & E-spARCH model (see \eqref{eq:spARCHy} and \eqref{eq:EspARCHh})                                                  \\
  \fct{qml.spARCH}    & \code{y $\sim$ 1}        & \code{"exp"}      & E-spARCH model with an additional intercept for the mean equation                                                                      \\
  \fct{qml.spARCH}    & \code{y $\sim$ a + b}    & \code{"exp"}      & Linear Regression with regressors \code{a} and \code{b} and E-spARCH residuals                                       \\
  \fct{qml.spARCH}    & \code{y $\sim$ a + b:c}  & \code{"exp"}      & Linear Regression with more complex expressions and E-spARCH residuals   \\
  \fct{qml.SARspARCH} & \code{y $\sim$ 0}        & \code{"gaussian"} & SAR model without an intercept, but with spARCH residuals  (see \eqref{eq:SARspARCH1} and \eqref{eq:SARspARCH2})                                                      \\
  \fct{qml.SARspARCH} & \code{y $\sim$ 1}        & \code{"gaussian"} & SAR model with an intercept and spARCH residuals                                                                                  \\
  \fct{qml.SARspARCH} & \code{y $\sim$ a + b}    & \code{"gaussian"} & SAR model with an intercept and the regressors \code{a} and \code{b}  and spARCH residuals                                        \\
  \fct{qml.SARspARCH} & \code{y $\sim$ a + b:c}  & \code{"gaussian"} & SAR model with more complex expressions and spARCH residuals    \\
  \fct{qml.SARspARCH} & \code{y $\sim$ 0}        & \code{"exp"}      & SAR model without an intercept, but with E-spARCH residuals  (see \eqref{eq:SARspARCH1} and \eqref{eq:SARspARCH2})                                                      \\
  \fct{qml.SARspARCH} & \code{y $\sim$ 1}        & \code{"exp"}      & SAR model with an intercept and E-spARCH residuals                                                                                  \\
  \fct{qml.SARspARCH} & \code{y $\sim$ a + b}    & \code{"exp"}      & SAR model with an intercept and the regressors \code{a} and \code{b}  plus E-spARCH residuals                                        \\
  \fct{qml.SARspARCH} & \code{y $\sim$ a + b:c}  & \code{"exp"}      & SAR model with more complex expressions and E-spARCH residuals    \\
  \hline
  \end{tabular}
  \caption{Overview of spatial models, which can be fitted by \fct{qml.spARCH} and \fct{qml.SARspARCH}.}\label{table:fittable_models}
\end{table}

\section{Real-data example: prostate cancer incidence rates}\label{sec:application}



Below, the focus is on the incidence rates (2008--2012) for prostate cancer provided by the Centers for Disease Control and Prevention (\citealt{CDC_data}). In particular, we analyze the incidence rates in all counties of several southeastern U.S. states, namely Arkansas, Louisiana, Mississippi, Tennessee, North and South Carolina, Georgia, Alabama, and Florida. This area also covers the counties along the Mississippi River collectively known as ``cancer alley'' (see \citealt{Nitzkin92,Brent10,Berry03}). All rates are age-adjusted to the 2000 U.S. standard population (cf. \citealt{CDC_data}).

As explanatory variables, we included a large set of environmental, climate, behavioral, and health covariates, which might have an influence on incidence rates for prostate cancer. For instance, we consider air pollution, such as $PM_{2.5}$, $PM_{10}$, $SO_2$, $NO_2$, $CO$, $O_3$, and $CH2O$, as potential environmental hazard factors. Moreover, we account for smoking, drinking, sport activities, and further healthcare-related variables as potential influences on the cancer incidence rates. In total, we account for 34 explanatory variables, which were obtained by inverse-distance-kriging from spatial points processes. Most of the variables are correlated, so we performed a factor analysis on 5 subgroups to identify 10 common factors. The factor loadings are summarized in Table \ref{table:factor_loadings}. Eventually, the final explanatory factors were chosen by minimizing the Bayesian information criterion using the generic function \fct{step} as follows.
\begin{lstlisting}
R> out <- step(qml.SARspARCH(formula, B = B, W = W, type = "gaussian"), k = log(length(Y)))
\end{lstlisting}
The \code{formula} object simply defines a linear model between the logarithmic incidence rates and all factors. Further, matrix $\xmat{B}$ describes the predefined spatial dependence structure in the mean equation. For this analysis, $\xmat{B}$ has been chosen as a row-standardized contiguity matrix of the direct neighbors. For the spatial dependence in the spatial ARCH term of the residuals, we also included all neighbors up to order 4. Hence, $\xmat{W}$ is the row-standardized matrix of the sum of the first-, second-, third-, and fourth-lag neighbors.

By minimizing the BIC criterion, the $2^{\text{nd}}$ and $10^{\text{th}}$ factor has been selected. Whereas the $2^{\text{nd}}$ factor has positive loadings mainly for fine particulate matters, $PM_{2.5}$ and $PM_{10}$, the $10^{\text{th}}$ describes the tendency for high blood pressure and cholesterol in the county’s population. However, note that this analysis is based on aggregated data rather than individual patients; hence, the selected factors cannot be interpreted as carcinogenic factors.

Using the generic \fct{summary} for the \class{spARCH} class, the estimated model can be summarized as follows.
\begin{lstlisting}
 Call:
qml.SARspARCH(formula = formula, B = B, W = W, type = "gaussian",
    data = NULL)

 Residuals:
      Min.    1st Qu.     Median       Mean    3rd Qu.       Max.
-0.7492270 -0.1079639 -0.0001509 -0.0005261  0.1121190  0.6404564

 Coefficients:
                        Estimate Std. Error t value  Pr(>|t|)
alpha (spARCH)         0.0203839  0.0042674  4.7766 1.783e-06 ***
rho (spARCH)           0.3782104  0.1309656  2.8879  0.003879 **
lambda (SAR)           0.6768133  0.0356765 18.9708 < 2.2e-16 ***
(Intercept)            1.5388985  0.1702222  9.0405 < 2.2e-16 ***
X_factor_scores[, 2]   0.0192857  0.0069917  2.7584  0.005809 **
X_factor_scores[, 10] -0.0205693  0.0064450 -3.1915  0.001415 **
---
Signif. codes:  0 *** 0.001 ** 0.01 * 0.05 . 0.1   1

 AIC: -1734.2, BIC: -1746.2 (Log-Likelihood: 873.11)

 Moran's I (residuals): -0.022899, p-value: 0.32023

 Moran's I (squared residuals): 0.021409, p-value: 0.00050052
\end{lstlisting}

First, we see that the model has a significant spatial autocorrelation in the mean equation since $\hat{\lambda}$ (\code{lambda (SAR)}) differs significantly from zero. This implies that there are clusters of higher prostate cancer incidence rates and, vice versa, lower incidence rates. Second, the error process shows conditional, autoregressive heteroscedasticity in space, which is captured by the spARCH component of the model, i.e., $\hat{\rho} = 0.378$ and $\hat{\alpha} = 0.020$. This can be interpreted as differences in the local uncertainty of the model. Hence, there are regions where the model predicts the true incidence rates more accurately, and there are regions with a worse fit. This can also be interpreted as local risks coming from unobserved, hidden factors. Note additionally that it is important to account for spatial conditional heteroscedasticity, as the estimates of spatial autoregressive models are biased if the error variance is not homogeneous across space. Inspecting the residuals, one can see that the spatial autocorrelation has been fully captured by the model, as Moran's $I$ of the residuals is close to zero. In contrast, there is a weak spatial dependence in the squared residuals. To inspect the reason for this dependence graphically, the function \fct{plot} can be used to produce the plots shown in Figure \ref{fig:output_emp_model}.

After fitting the model, one also may include further regressors or estimate an intercept-only model via \fct{update}. For illustration, we added the percentage of positive results for a prostate-specific antigen (PSA) test in each county as an additional explanatory variable by
\begin{lstlisting}
R> out2 <- update(out, . ~ . + PSA_test)
\end{lstlisting}
The PSA test is used for prostate cancer screening, meaning that there should definitely be a positive dependence between the PSA test and the incidence rates. In fact, the estimated parameter is positive, and the AIC is lower compared to the previous model. To be precise, the updated parameters are
\begin{lstlisting}
                        Estimate Std. Error t value  Pr(>|t|)
alpha (spARCH)         0.0199281  0.0043105  4.6231  3.78e-06 ***
rho (spARCH)           0.3902185  0.1280266  3.0479 0.0023041 **
lambda (SAR)           0.6643605  0.0366748 18.1149 < 2.2e-16 ***
(Intercept)            1.1349551  0.2301554  4.9313  8.17e-07 ***
X_factor_scores[, 2]   0.0198504  0.0069903  2.8397 0.0045159 **
X_factor_scores[, 10] -0.0224035  0.0065828 -3.4034 0.0006656 ***
PSA_test               0.0095962  0.0042728  2.2459 0.0247125 *
\end{lstlisting}
%
%
%

\begin{figure}
  \centering
  \includegraphics[width=\textwidth]{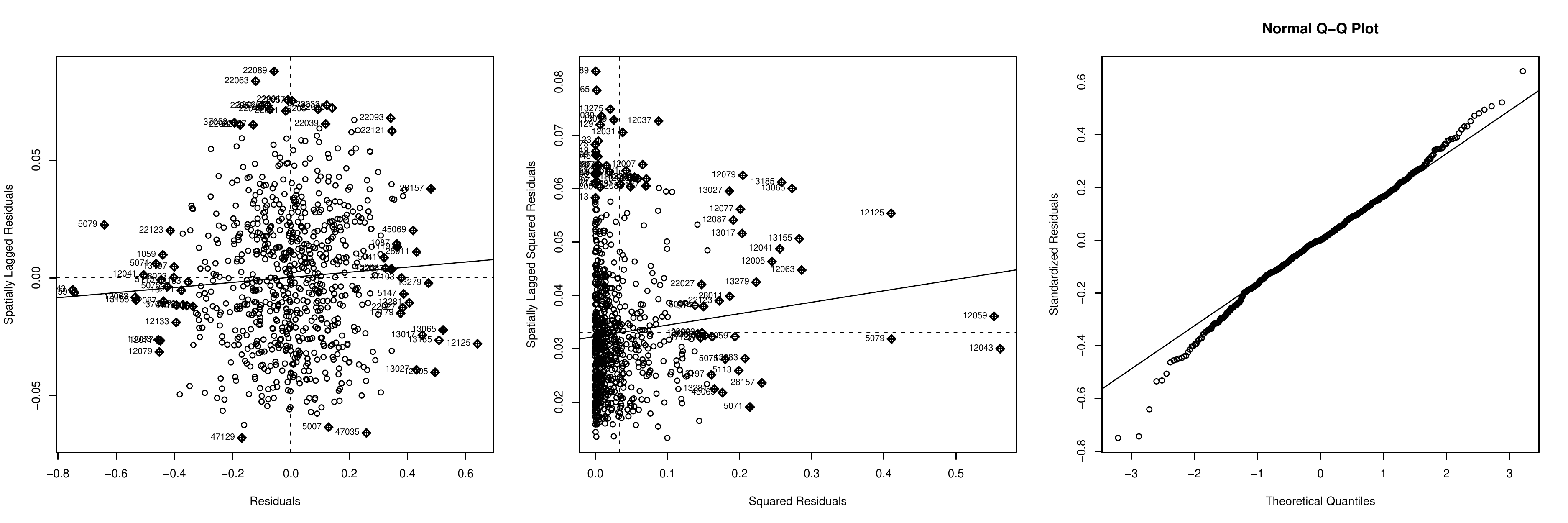}
  \caption{Resulting graphical output of \fct{plot} for the real-data example.}\label{fig:output_emp_model}
\end{figure}

\begin{table}
\centering
\begin{scriptsize}
\begin{tabular}{l rrrrrrrrrr}
  \hline
 & F. 1 & F. 2 & F. 3 & F. 4 & F. 5 & F. 6 & F. 7 & F. 8 & F. 9 & F. 10 \\
  \hline
  $PM_{2.5}$ concentration & 0.69 & 0.72 &  &  &  &  &  &  &  &  \\
  $SO_2$ concentration & 0.33 & -0.03 &  &  &  &  &  &  &  &  \\
  $NO_2$ concentration & 0.13 & -0.12 &  &  &  &  &  &  &  &  \\
  $CO$ concentration & 0.31 & 0.05 &  &  &  &  &  &  &  &  \\
  $PM_{10}$ concentration & 0.07 & 0.44 &  &  &  &  &  &  &  &  \\
  $O_3$ concentration & 1.00 & -0.02 &  &  &  &  &  &  &  &  \\
  Solar radiation &  &  & 0.60 & 0.44 &  &  &  &  &  &  \\
  Precipitation &  &  & -0.08 & -0.26 &  &  &  &  &  &  \\
  Outdoor temperature &  &  & 1.00 & -0.05 &  &  &  &  &  &  \\
  Temperature differences &  &  & 0.32 & 0.94 &  &  &  &  &  &  \\
  Ambient maximal temperature &  &  & 0.08 & -0.39 &  &  &  &  &  &  \\
  $CH_2O$ & -0.23 & 0.32 &  &  &  &  &  &  &  &  \\
  Percentage of current smokers &  &  &  &  & 0.47 & -0.85 &  &  &  &  \\
  Percentage of former smokers &  &  &  &  & 0.92 & 0.37 &  &  &  &  \\
  Smoke some days &  &  &  &  & -0.07 & -0.62 &  &  &  &  \\
  Never smoked &  &  &  &  & -0.96 & 0.25 &  &  &  &  \\
  Aerobic activity  &  &  &  &  &  &  & -0.05 & 0.58 &  &  \\
  Exercises &  &  &  &  &  &  & 0.41 & 0.33 &  &  \\
  Physical activity index &  &  &  &  &  &  & -0.09 & 0.99 &  &  \\
  Alcohol consumption &  &  &  &  & 0.04 & 0.62 &  &  &  &  \\
  Binge drinking &  &  &  &  & 0.07 & 0.44 &  &  &  &  \\
  Heavy drinking &  &  &  &  & 0.43 & 0.02 &  &  &  &  \\
  High cholesterol  &  &  &  &  &  &  &  &  & 0.00 & 1.00 \\
  Cholesterol checked &  &  &  &  &  &  &  &  & 0.55 & 0.00 \\
  Overweight (BMI 25.0-29.9) &  &  &  &  &  &  &  &  & 0.99 & 0.09 \\
  Obese (BMI 30.0 - 99.8) &  &  &  &  &  &  &  &  & -0.75 & 0.01 \\
  Blood stool test &  &  &  &  &  &  &  &  & 0.56 & -0.23 \\
  Sigmoidoscopy &  &  &  &  &  &  &  &  & 0.14 & -0.16 \\
  High blood pressure &  &  &  &  &  &  &  &  & 0.03 & 0.79 \\
  Flu shot &  &  &  &  &  &  & 0.81 & -0.13 &  &  \\
  Pneumonia vaccination &  &  &  &  &  &  & 0.51 & -0.26 &  &  \\
  Health care coverage &  &  &  &  &  &  & 0.58 & 0.18 &  &  \\
  Seatbelt use &  &  &  &  &  &  & -0.58 & 0.10 &  &  \\
   \hline
\end{tabular}
\end{scriptsize}
\caption{Overview of all included regressors and factor loading for the 10 common factors. The regressors were divided into 5 subgroups to allow for distinctions between the factors.}\label{table:factor_loadings}
\end{table}

%

%
%
%

\section{Summary and discussion}\label{sec:summary}

This paper examines spatial models for autoregressive conditional heteroscedasticity. In contrast to previously proposed spatial GARCH models, these models allow for instantaneous autoregressive dependence in the second conditional moments. Previous approaches only allowed for spatial dependence in the first temporal lag. However, these models are also captured by the spatial ARCH approach, since temporal dependence can be included by appropriate choices of the weighting matrix. In addition to discussing previously proposed models, we introduced a novel spatial exponential ARCH model, for which the probability density has been derived and maximum-likelihood estimators discussed.

In addition to this theoretical model, we focus on the computational implementation of all considered spatial ARCH models in the \proglang{R}-package \pkg{spGARCH}. In particular, the simulation and estimation has been demonstrated. Regarding maximum-likelihood estimation, a broad range of spatial models are implemented in the package. Furthermore, the spatial weights matrices, as well as the mean model, can easily be specified by the user, providing a flexible and easy-to-use tool for spatial ARCH models. All estimation functions return an object for class \class{spARCH}, for which several generic functions are provided, such as \fct{summary}, \fct{plot}, and \fct{AIC}. This setup also allows the use of the \proglang{R}-base functions, such as \fct{step} for stepwise model selection or \fct{update} for updating the results of different mean models. Eventually, the use of these functions are demonstrated by an empirical example, namely county-level incidence rates of prostate cancer.

In the future, the package should be extended for further spatial ARCH-type models. Along this vein, a class for model specifications should be added alongside the actual implementations via arguments for the fitting functions. In that way, the package can be aligned to common time series ARCH packages, such as the \pkg{rugarch} package. Furthermore, the package could benefit from robust estimation methods, another focus for future research.

\section{Appendix}

\begin{proof}[Proof of Theorem \ref{th:EspARCH}]
For this definition of $g_b$, one could rewrite $\ln \xvec{h}$ as
\begin{equation}\label{eq:EspARCHh2}
\ln \xvec{h}_E = \xmat{S} \left(\alpha\xvec{1} + \rho b \xmat{W} \ln|\xvec{Y}|\right)
\end{equation}
with
\begin{equation*}
\xmat{S} = (s_{ij})_{i,j = 1, \ldots, n} = \left(\xmat{I} + \frac{1}{2} \rho b \xmat{W}\right)^{-1} \, .
\end{equation*}
Since $w_{ij} \geq 0$ for all $i,j = 1, \ldots, n$, $\xmat{W}$ is positive definite and it holds that
\begin{equation*}
\det \left(\xmat{I} + \frac{1}{2} \rho b \xmat{W}\right) \geq 1 + \frac{1}{2} \rho b \det(\xmat{W}) > 0 \, .
\end{equation*}
Thus, the relation between $Y(\xvec{s}_1), \ldots, Y(\xvec{s}_n)$ and $\varepsilon(\xvec{s}_1), \ldots, \varepsilon(\xvec{s}_n)$ is given by \eqref{eq:spARCHy} and \eqref{eq:EspARCHh2}.
\end{proof}

\begin{proof}[Proof of Corollary \ref{cor:EspARCH}]
For $\rho \geq 0$, $b \geq 0$, and $w_{ij} \geq 0$ for all $i,j$, the inverse
\begin{equation*}
\xmat{S} = (s_{ij})_{i,j = 1, \ldots, n} = \left(\xmat{I} + \frac{1}{2} \rho b \xmat{W}\right)^{-1} \, .
\end{equation*}
is a non-negative matrix. Thus,
\begin{equation*}
\ln \xvec{h}_E = \xmat{S} \left(\alpha\xvec{1} + \rho b \xmat{W} \ln|\xvec{Y}|\right)
\end{equation*}
is positive for $\alpha > 0$.
\end{proof}

\end{document}